\documentclass[11pt]{article}
\usepackage{graphicx}
\usepackage{times}
\usepackage{fullpage}
\usepackage{epsf}
\usepackage{color}
\usepackage{xcolor}
\usepackage{amsmath,amsfonts}
\usepackage{graphics} 
\usepackage{hyperref}
\usepackage{amsmath,amssymb,booktabs,tabularx,array,microtype}
\definecolor{forestgreen}{rgb}{0.13, 0.55, 0.13}
\definecolor{bananayellow}{rgb}{1.0, 0.88, 0.21}

\usepackage{lineno}
\usepackage{algorithmic}
\usepackage[ruled,vlined,commentsnumbered,titlenotnumbered]{algorithm2e}

\def\showlabel#1{}

\newfont{\mycrnotice}{ptmr8t at 7pt}
\newfont{\myconfname}{ptmri8t at 7pt}

\def\maxringsize{{18}}

\SetCommentSty{mycommfont}
\SetKwBlock{Loop}{loop}{}

\newcommand\cC{{\mathcal C}}
\newcommand\cZ{{\mathcal Z}}
\newcommand\cT{{\mathcal T}}
\newcommand\cR{{\mathcal R}}
\newcommand\cK{{\mathcal K}}
\newcommand\cKC{{\mathcal KC}}
\newcommand\cD{{\mathcal D}}

\renewcommand{\hat}{\widehat}

\newtheorem{proposition}{Proposition}
\newtheorem{lemma}[proposition]{Lemma}

\newtheorem{theorem}[proposition]{Theorem}

\newtheorem{task}[proposition]{Task}

\newcommand{\qed}{\hbox{\rule{6pt}{6pt}}}
\newenvironment{proof}[1][]{\paragraph{Proof{#1}}}{\hfill\qed\medskip\\}

\newcommand\boundary\partial

\newcommand\drop[1]{}

\newcommand\req[1]{(\ref{#1})}

\newcounter{invari}
\begin{document}

\title{Four-Coloring Planar Graphs in Linear Time}

\author{
    {\em Yuta Inoue}\thanks{The University of Tokyo, Tokyo, Japan, \texttt{yutainoue@is.s.u-tokyo.ac.jp}.\\  
    Supported by JSPS Kakenhi 26K21777 and JP25K24465, by JST ASPIRE JPMJAP2302, and by Hirose Foundation Scholarship.}
    \and
    {\em Ken-ichi Kawarabayashi}\thanks{National Institute of Informatics \& The University of Tokyo, Tokyo, Japan, \texttt{k\_keniti@nii.ac.jp}.\\ 
    Supported by Alexander von Humboldt Foundation (Humboldt Research Award), JSPS Kakenhi 26K21777 and JP25K24465, and by JST ASPIRE JPMJAP2302.}
    \and
    {\em Atsuyuki Miyashita}\thanks{The University of Tokyo, Tokyo, Japan, \texttt{miyashita-atsuyuki869@is.s.u-tokyo.ac.jp}.\\  
    Supported by JSPS Kakenhi 26K21777 and JP25K24465, by JST ASPIRE JPMJAP2302, and by JSPS KAKENHI 26KJ0817.}
    \and
    {\em Bojan Mohar}\thanks{Simon Fraser University, Burnaby, BC, Canada \& FMF, University of Ljubljana, Slovenia, \texttt{mohar@sfu.ca}.\\ 
    Supported in part by NSERC Discovery Grant R832714 (Canada),
    by the ERC Synergy grant (European Union, ERC, KARST, project number 101071836),
    and by the Research Project N1-0218 of ARIS (Slovenia).}
    \and
    {\em Carsten Thomassen}\thanks{Technical University of Denmark, \texttt{ctho@dtu.dk}}
    \and
     {\em Mikkel Thorup}\thanks{ University of Copenhagen, Denmark, 
 \texttt{mikkel2thorup@gmail.com}.\\ 
 Supported by VILLUM Foundation grant 54451, Basic Algorithms Research Copenhagen (BARC).}}

\maketitle

\begin{abstract}
The famous Four Color Theorem (4CT) states that any planar graph can be colored with four colors, but how do we find a 4-coloring? In our previous paper [FOCS'26], we presented a simple combinatorial algorithm to 4-color planar graphs in $O(n\log n)$ time. We improve it here to $O(n)$ time, which is optimal.

The algorithm from our previous paper is based on proving
that every planar triangulation has a linear number of reducible configurations or obstructing cycles. 
The bottleneck for the algorithm is that we need $O(\log n)$ Kempe changes to extend a recursive 4-coloring to the reducible configurations. Each Kempe change can involve a constant fraction of the graph and takes linear time.

Here we show that 
all $O(\log n)$ Kempe changes can be identified and implemented in $O(n)$ time. The idea is to spend $O(n)$ time to preprocess the graph so that we can identify and simulate each good Kempe change in $O(n/\log n)$ amortized time. Thus, we implement all $O(\log n)$ Kempe changes in $O(n)$ total time. Based on this, we implement the whole 4-coloring in linear time. 

We note that creating such a representation for fast identification and simulation of the right Kempe changes is a challenge of data structures, unlike the previous paper that builds on
a new graph theoretic understanding of reducible configurations.
\end{abstract}
\newpage

\section{Introduction}
The celebrated Four Color Theorem (4CT) 
was conjectured by Francis Guthrie in 1852 (see \cite{madison1897,tait1880dual} for some of the early history), and remained open for more than 100 years, until Appel and Haken found a proof \cite{4ct1, 4ct2} in 1976. The theorem states that every planar map or graph with no loops can be colored with four colors. So 4-coloring is always possible, but how do we 4-color a concrete planar graph? There is no easy answer. In fact, historically, it has been so complicated that it has led to many proposed counterexamples to 4CT.

Robertson, Sanders, Seymour, and Thomas \cite{RSST, RSST-STOC} 
found a simpler proof and
proved in 1997 that a 4-coloring can be found in quadratic time.

In a preceding paper \cite{nearlinear}, we proved that 4-coloring can be done in $O(n\log n)$ time. In this paper, we show how to do it in linear time, that is,
\begin{theorem}\label{thm:linear-time}
    We can 4-color a planar graph of order $n$ in $O(n)$ time.
\end{theorem}
This optimal time bound builds on the graph-theoretic results from
\cite{nearlinear} which yielded a combinatorial $O(n\log n)$ time solution. Our improvement to $O(n)$ time is based on a new efficient data structure. More precisely, we use the word RAM model with word length $\Theta(\log n)$. The unit-cost operations are the standard AC$^0$ operations available in an imperative programming language such as C \cite{kernighan1978c}: direct and indirect addressing, conditional jump, addition, bitwise Boolean operations, and shifts (called the Practical RAM in  \cite{miltersen1996lower}).

\subsection{Graph theoretic foundations}
In order to describe our contribution in more detail, we first review the relevant concepts and results from previous work on 4CT. If we just say \emph{coloring}, we mean a proper $4$-coloring with colors 1, 2, 3, 4.

By a \emph{planar triangulation} (or just \emph{triangulation}) we mean a (simple) graph embedded in the 2-sphere, all of whose faces are triangles. The 4CT says that any planar graph can be 4-colored, but since triangulations can only make coloring harder, it suffices to consider planar triangulations. 
We say a cycle $C$ is \emph{obstructing} if it is of length at most 4 with at least one vertex on each side, or of length 5, with at least 2 vertices on each side.
Birkhoff proved in \cite{birkhoff} that a minimal counterexample to 4CT has no obstructing cycles.

\paragraph{Kempe chains and Kempe changes}

Given a 4-colored graph $G$ and two colors $a,b$,
an \emph{$ab$-colored component} is a component of the subgraph of $G$ induced by the vertices of color $a$ or $b$.
To define Kempe chains and changes, we first need to decide a \emph{color pairing} $\{a,b\}$ and $\{c,d\}$ with $\{a,b,c,d\}=\{1,2,3,4\}$. The \emph{Kempe chains}
are now the $ab$- and $cd$-colored components. For a \emph{Kempe change}, we swap the colors on a subset of the Kempe chains; that is, for an $ab$-colored component, a swap switches $a$ and $b$, and on a $cd$-colored component, it switches $c$ and $d$. Such a Kempe change preserves a proper 4-coloring. The use of Kempe changes was pioneered by Kempe in \cite{kempe1879}, and has been used in all later correct proofs of 4CT.

\paragraph{Configurations}
A key concept from the classic proof of the 4CT \cite{AppelHaken89,RSST} is that of a reducible configuration. Roughly speaking, a reducible configuration is a locally defined configuration such that, if it appears in a triangulation, it cannot be a minimal counterexample to 4CT.
To define reducible configurations formally, we first need to define configurations.

Sometimes we view the triangulation $G$ we want to color as embedded in the plane, and we may or may not treat one of the faces as the outer face. A \emph{near-triangulation} is a connected graph embedded in the plane such that all facial walks, except for the outer face, are triangles. 
This includes the \emph{trivial near-triangulations} consisting of a single edge or vertex.

By a \emph{configuration} we mean a pair $(Z,\delta)$, where $Z$ is a near-triangulation and $\delta: V(Z)\to {\mathbb N}$ is the \emph{degree function} satisfying the following condition: if $v\in V(Z)$ is a vertex that is not on the boundary of $Z$, then $\delta(v)=d_Z(v)$; otherwise $\delta(v)>d_Z(v) $. Moreover, if $v$
is a cut-vertex of $Z$, then
$\delta(v)=d_Z(v)+2$.

We say that a plane triangulation $G$ \emph{contains the configuration $(Z,\delta)$} if $Z$ is embedded in $G$ such that every inner triangular face of $Z$ is also a face in $G$, and all triangular faces of $Z$ have the same orientation in $G$. Moreover, we require that the degree of each vertex $v\in V(Z)$ in $G$ is equal to $\delta(v)$, that is $d_G(v)=\delta(v)$. Finally, if $v$ is a cut vertex in $Z$, its two neighbors outside $Z$ in $G$ are non-consecutive around $v$.
We say that $Z$, or the configuration $(Z,\delta)$, is \emph{induced} in $G$ if $Z$ is an induced subgraph of $G$.

\begin{figure}[htb]
    \centering
    \includegraphics[width=0.62\textwidth]{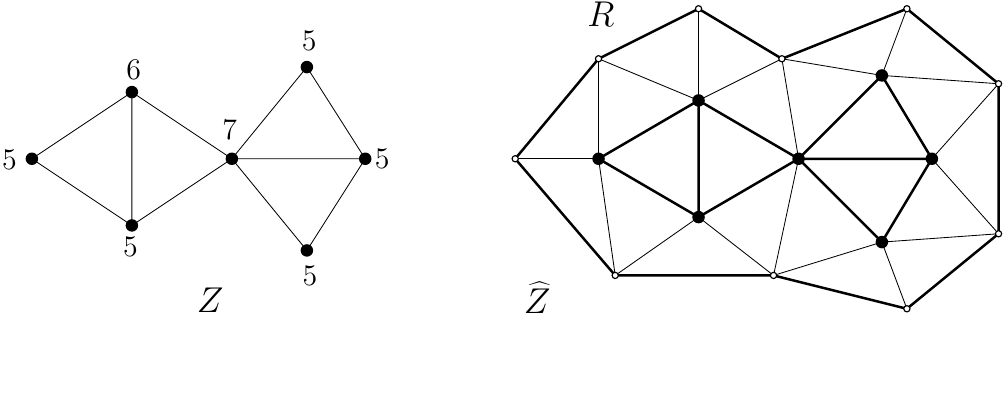}
    \caption{A D-reducible configuration $(Z,\delta)$ with a cut-vertex and its free completion.
    }
    \label{fig:extended configuration}
\end{figure}

If $(Z,\delta)$ is a configuration, the \emph{free completion of $(Z,\delta)$} is the near-triangulation $\widehat Z$ obtained from $Z$ by adding a cycle $R=R(Z)$, disjoint from $Z$ and joining vertices on the outer face of $Z$ to $R$ in such a way that $R$ becomes the outer face of the near-triangulation $\widehat Z$ and such that $\widehat Z$ contains $(Z,\delta)$. See the illustration in Figure \ref{fig:extended configuration}.

The length of the cycle $R$ is precisely
$$
   |R| = \sum_{v\in V(Z)} (\delta(v) - d_Z(v)) - t = \sum_{v\in V(Z)} \delta(v) - 2|E(Z)| - t,
$$ 
where $t$ is the length of the outer facial walk of $Z$. The cycle $R = R(Z)$ is called the \emph{ring} of the free completion.

\paragraph{Reducible configurations}
Suppose that $G$ contains an induced configuration $(Z,\delta)$. Since $Z$ is induced in $G$, the triangulation $G$ contains the free completion of $Z$, possibly with some vertices and edges on $R(Z)$ that may be identified in $G$. Equivalently, when $Z$ is induced, the embedding $\phi$ of $Z$ into $G$ can be extended to a homomorphism $\phi$ from $\hat Z$ into $G$, but this homomorphism may not be 1-1 on $R(Z)$.

Every 4-coloring $\psi_0$ of $G-V(Z)$ induces a 4-coloring 
of the ring $R=R(Z)$ in the free completion. If this 4-coloring can be extended to a 4-coloring $\psi_1$ of $\widehat Z$, then the combination of $\psi_0$ and $\psi_1$ gives a 4-coloring of $G$. In the proof of 4CT, special configurations are shown to be reducible. This happens if any 4-coloring $\psi_0$ can be changed (by using a sequence of Kempe changes in $G-V(Z)$) so that the resulting coloring extends from $R$ to $\widehat Z$. We describe the precise definition and the process related to Kempe changes next.

We say that a coloring $\psi$ of $R$ is \emph{$0$-extendible} (or just \emph{extendible}) to $Z$ if it can be extended to a coloring of $\widehat Z$. 
For each integer $i>0$ we say that $\psi$ is \emph{$i$-extendible} to $Z$ if it is either $(i-1)$-extendible or the following holds: If $G'$ is an arbitrary plane graph containing the ring $R$ as a facial cycle and $\psi$ is any 4-coloring of $G'$ extending $\psi$, then there is a Kempe change of $\psi$ such that the resulting
4-coloring on $R$ after making the Kempe change is an $(i-1)$-extendible 4-coloring of $R$. We define the \emph{reducibility level} (or just \emph{level}) of $\psi$ as the minimum $i$ such that $\psi$ is $i$-extendible.
A Kempe change in $G-V(Z)$ is called an \emph{improving Kempe change} if it changes the coloring of $R$ to a lower level.

We say that the configuration $(Z,\delta)$ is D$^i$-\emph{reducible} if every 4-coloring of its ring is $i$-extendible, and D-\emph{reducible} means D$^i$-reducible for some finite $i$. 
Exercising the definition, we have

\begin{lemma}\label{lem:3-4}
    A single vertex of degree 3 or 4 is D-reducible. It is 0-extendible if the degree is 3 and 1-extendible if the degree is 4.
\end{lemma}

\begin{proof}
    If the vertex has degree 3, we can remove it, color the remaining graph, and then put it back, assigning a color not used in its ring.
    If the vertex $v$ has degree 4, we also remove it and color the rest. If the ring has at most three
    colors, we are done as above. Otherwise, we need a Kempe change. Let $v_0,\ldots, v_3$ be the ring vertices in the order they appear 
    around $v$. Look at the Kempe chain $A$ of $v_0$ induced by the colors of $v_0$ and $v_2$. If $A$ does not include $v_2$, we simply perform the Kempe change on $A$. Otherwise, $A$ connects $v_0$ and $v_2$, so we can perform a Kempe change on $v_1$ to give it the same color as $v_3$. In each case, the modified coloring can be extended to $v$. The proof for the degree-4 case goes back to Kempe \cite{kempe1879}.
\end{proof}
It is easy to see \cite{AllaireSwart78,Franklin1922} that a minimum counterexample to 4CT cannot contain D-reducible configurations.

Note that a D-reducible configuration $(Z,\delta)$ may be contained as a non-induced subgraph of a triangulation $G$. In that case we may not be able to use Kempe-changes to extend a 4-coloring of $G-V(Z)$ to $Z$. Nevertheless, we still say that the subgraph $Z$ in $G$ is D-reducible.

We note that many works on 4CT, including the first proof in \cite{AppelHaken89}, use a more general notion of reducibility, called C-reducibility introduced by Franklin \cite{Franklin1922} and used in \cite{AppelHaken89} and \cite{RSST}; see also \cite{AllaireSwart78}. With C-reducibility, we do not just remove $V(Z)$ from $G$, but we also do some contractions on the resulting face. The proofs of 4CT in 
\cite{AppelHaken89,RSST} show that any triangulation contains either an obstructing cycle or a D- or C-reducible configuration from a given finite set of such configurations. However, \cite{steinberger2010unavoidable} proved that a set of D-reducible configurations suffices, which will be important to us.

In this paper, we mostly discuss configurations $(Z,\delta)$ with a given embedding in $G$, where for each vertex $v\in V(Z)$, $\delta(v)=d_G(v)$. Therefore, we will often
identify such a configuration with the embedded near-triangulation $Z$.

\paragraph{A quadratic algorithm.}
In \cite{RSST-STOC}, the authors derive a quadratic-time 4-coloring algorithm. 

The proofs above imply that for any triangulation, we can find either an obstructing cycle or a reducible configuration in linear time.

If we find a reducible configuration $Z$, we reduce it and color the rest by induction. This gives a coloring of the ring $R$. If it is an $i$-extendible configuration, then we may have to perform up to $i$ improving Kempe changes before we can extend the coloring to the original $Z$. 

To tell if a Kempe chain is improving, we need to know the level of all colorings of the ring $R$. This is only a constant amount of information, since we only have a constant number of reducible configurations of constant size.

Also note that we only know that an improving Kempe change exists, but which one works depends on the graph and the coloring outside the ring. Recall that a Kempe change is characterized by color pairing $\{a,b\}$ and $\{c,d\}$ of colors and a subset of $ab$- and $cd$-colored Kempe chains on which the two colors are swapped. The swapping only affects the coloring of the ring if the Kempe chain
intersects the ring. This means that there are fewer than $3\cdot 2^{|R|}=O(1)$ relevant Kempe changes to check, and we apply the first improving one we find. We can implement all of this in linear total time.

\cite{RSST-STOC} also shows how we can benefit from obstructing cycles. The general conclusion is that we can reduce the problem size by an additive constant in linear time, so the overall algorithm is quadratic. 

The main bottleneck in the above approach is the Kempe changes. The interaction between different Kempe chains can be very complicated. Each Kempe chain can be of linear size, and if, for example, we first perform a Kempe change using color pairs $a, b$ and $c, d$, then it can affect the $ac$, $ad$, $bc$, and $bd$-colored Kempe chains. So, multiple Kempe changes can be quite chaotic. It seems hopeless to update the information about Kempe chains in sublinear time.

\subsection{Our previous work}\label{sec:near-linear}

As mentioned earlier, our optimal $O(n)$ coloring is an efficient implementation of our previous combinatorial $O(n\log n)$ algorithm from \cite{nearlinear}. In the following we review the relevant parts of \cite{nearlinear}.
Let us recall that two subgraphs $X$ and $Y$ of $G$ \emph{touch} if either $V(X)\cap V(Y)\ne \varnothing$, or there is an edge $xy\in E(G)$ with $x\in V(X)$ and $y\in V(Y)$.

\paragraph{Linear number of reducible configurations}
The main result from our preceding paper \cite{nearlinear} is as follows.

\begin{theorem}[{\cite{nearlinear}}]\label{thm:main-linear}\showlabel{thm:main-linear}
There is a linear-time algorithm that, given any triangulation $G$, finds one of the following: \begin{itemize}
    \item[(a)] A linear number of pairwise non-touching induced D-reducible configurations in $G$, each of which from a given set $\cD$ of size 8202. 
    \item[(b)] A linear number of non-touching obstructing cycles in $G$.
    \end{itemize}
        All vertices in the above-mentioned reducible configurations or obstructing cycles have degree at most 12 in $G$. Moreover, for the reducible configurations in $\cD$, the maximal ring size is
        18 and the maximal reducibility level is 25.
\end{theorem}
Note that prior to \cite{nearlinear} there was not even an existential proof that a triangulation had more than one reducible configuration or obstructing cycle (which does suffice for 4CT).

\paragraph{An $O(n\log n)$ recursion}
Based on Theorem \ref{thm:main-linear}, \cite{nearlinear} derived an $O(n\log n)$ time 4-coloring algorithm for a planar triangulation $G$ of order $n$. For this, it suffices that we can construct in $O(n\log n)$ time a planar triangulation $H$ that is smaller by a constant factor, color $H$ recursively, and use this coloring to color $G$ in $O(n\log n)$ additional time.

In case (b) of Theorem \ref{thm:main-linear}, where we get a linear number of non-touching obstructing cycles, \cite{nearlinear} did the recursive step in linear time. The bottleneck in \cite{nearlinear} is case (a) where the recursive step takes $O(n\log n)$ time. Below we discuss this case.

\paragraph{D-reducible configurations.}
Suppose as in (a) that  we get $\Omega(n)$ non-touching induced D-reducible configurations from $\cD$ in $G$. 
In the quadratic-time algorithm, we would just handle one of them in linear time, but
here we will handle all of them in $O(n\log n)$ time. 
Below, we first show this using randomization. 

More specifically, our first obvious step is to remove all the reducible configurations from the graph. This removes $\Omega(n)$ vertices. Next, using new auxiliary edges, we triangulate all the faces inside their now empty rings. Then we recursively color the resulting reduced triangulation, which has a constant factor fewer vertices.

Having colored the reduced
triangulation recursively, we restore the original triangulation exactly as it was, removing the auxiliary edges and putting back all the reducible configurations. However, we keep the recursive coloring, which is now a 4-coloring of all vertices outside the reducible configurations. The reducible configurations, while uncolored, have all their rings colored. This follows because they are non-touching, implying that the ring of one reducible configuration does not intersect any other reducible configuration.
Kempe chains and changes only apply to the colored vertices.

\paragraph{A randomized approach}

We will now start doing random Kempe changes on the colored vertices. If the ring $R$ of a non-colored configuration $Z$ gets a 0-extendible coloring, then we extend this 4-coloring of $R$ to $Z$, and then $Z$ is no longer among the non-colored configurations.

In general, if the coloring of a ring is  $i$-extendible, then we may have to
do $i$ improving Kempe changes to make it $0$-extendible. We know from Theorem \ref{thm:main-linear} that all colorings of the rings of our reducible configurations are 25-extendible. Each Kempe change may affect most of the graph, so we cannot afford to do this for one configuration at a time. 
However, we claim that random Kempe changes are expected to be improving for a constant fraction of 
the non-colored configurations, and this leads to the following very simple algorithm.

While there is a non-colored configuration, we repeat the following two
steps:
\begin{enumerate}
    \item Make
    a random Kempe change.
    First, we randomly
    decide the color pairing $ab$ and $cd$ that will be used by the Kempe chains, that is, the Kempe chains 
    considered are the $ab$- and $cd$-colored 
    components. These partition all colored vertices. Next, for every $ab$-colored and every $cd$-colored component, we toss a coin to decide whether to swap its colors in the random Kempe change.
    \item For all non-colored components $Z$ whose ring coloring became $0$-extendible, we extend the coloring to $Z$.
\end{enumerate}
We claim that this 
algorithm terminates after $O(\log n)$ iterations with high probability. To prove 
this, we consider an arbitrary non-colored configuration $Z$, and argue that its ring coloring will become $0$-extendible within $O(\log n)$ rounds with high probability. The result then follows by a union bound.

Consider the iterations from
the perspective of any given reducible configuration $Z$ with ring $R$. If $Z$ is non-colored, it wants a sequence of at most 25 improving Kempe changes. As discussed in the quadratic algorithm, an improving Kempe change requires the right color pairing $ab$ and $cd$.
There is a $1/3$ chance that a random choice is right since the pairs are determined just by the pairing of color 1. Now the improvement only depends on the color change (color swap or no-swap) 
of the disjoint $ab$ and $cd$-colored components (Kempe chains) intersecting the ring $R$, of which there are at most
$|R|\leq 18$. A random Kempe change is thus improving for $Z$ with probability at least $\tfrac{1}{3}\cdot 2^{-18}$. We need at most 25 such consecutive improvements for $Z$ for the ring to become $0$-extendible. The probability of this success for $Z$ is at least $p=1/(3\cdot 2^{18})^{25}=\Omega(1)$.

The above fails for $Z$ if we make a Kempe change that is non-improving for $Z$ before $R$ becomes $0$-extendible, but if it fails, we still have a coloring of $R$, and any coloring of $R$ is 25-extendible\footnote{This would not be the case for C-reducible configurations as used in \cite{AppelHaken89,RSST-STOC}.}. If we now do $r\times 25$ iterations, then the probability that we fail to
color $Z$ is at most $(1-p)^r<e^{-pr}$. With
$50p^{-1}\ln n$ iterations, this is at most $1/n^2$. A union bound implies that 
all configurations succeed within $50p^{-1}\ln n$ iterations with probability at least $1-1/n$.

\paragraph{Derandomization using method of conditional expectations}

Using the method of conditional expectations, we derandomize the random choice of Kempe changes above.
The derandomization always makes the greedy choice maximizing the expected number of configurations that become 0-extendible. In particular, it will always get at least one 0-extendible configuration within 25 Kempe changes.
We will now start doing Kempe changes on the colored vertices. If the ring $R$ of a non-colored configuration $Z$ has a 0-extendible coloring, we extend this 4-coloring of $R$ to $Z$, and $Z$ is no longer among the non-colored configurations (strictly speaking, $(Z,d_G)$ is the reducible configuration). 

The main step is to show that we can deterministically find a Kempe change that improves the level of coloring for a constant fraction of the rings.
The key is to prove the following lemma:

\begin{lemma}[\cite{nearlinear}]\label{lem:det-Kempe-change}
    Let $Z_1,\ldots, Z_q$ be some (not necessarily all) non-colored reducible configurations in a triangulation $G$ that is 4-colored except for the remaining non-colored configurations. The colorings of their rings  $R_1,\ldots,R_q$ are not yet $0$-extendible.
    
    We can then deterministically, in linear time, find and execute a Kempe change that improves at least a fraction $p=\frac1{3\cdot 2^{18}}$ of the rings, and also tell which rings got improved.
\end{lemma}
\begin{proof}
We assume that for every one of our reducible configurations $D \in \cD$ (c.f. Theorem \ref{thm:main-linear}), we store the Kempe level of each coloring of its ring. This is a huge constant amount of information, but we need it to determine if a given Kempe change will improve.

For each of the three color combinations, that is, for $a\in\{2,3,4\}$, we use the color pairs $\{1,a\}$ and $\{2,3,4\}\setminus \{a\}$, and we identify all the bi-colored Kempe chains. Each vertex now knows which Kempe chain it belongs to. 

We now consider the perspective of a given ring $R_i$. Its coloring is only affected by the swap or no-swap of the Kempe chains intersecting it, and there are at most $|R_i|\le18$ such Kempe chains. We now consider all combinations of swap and no-swap for each
intersecting Kempe chain. We know which vertices are involved, hence the effect on the coloring, so we can 
check if a Kempe change following a given combination of intersecting swap/no-swap is improving.

We repeat this for all three color combinations, and for each ring we know at least one color combination and a corresponding intersecting swap/no-swap combination improve the coloring; we pick one such combination for each ring $R_i$.

For the next Kempe change, if we picked both the color combination and the swap/no-swap for each Kempe chain uniformly at random, then the probability that the change was improving for $R_i$ would be at
least $1/(3\cdot 2^{|R_i|})\geq p=1/(3\cdot 2^{18})$, and hence we would expect a fraction $p$ of all rings to get improved. 

We can get an improvement for at least a fraction $p$ of the rings using the method of conditional expectation as described below.

At the beginning, all our rings are ``live". For each ring, we know the color pair combination for the improving Kempe change, and we pick the color pair that works for most rings. Only the rings that get their color combination stay alive; the rest will ``die". We know that at least 1/3 stay alive.

We now go through each Kempe chain using the selected color pairs and decide whether to swap the colors. A ring $R_i$ only stays live as long as all decided swaps were as it wanted them. If the number of undecided Kempe chains intersecting $R_i$ is $u_i$, and if all remaining choices were made at random, then the probability that $R_i$ would survive till the end is $2^{-u_i}$. Thus, conditioned on decisions made, the expected number of survivors is
\begin{equation}\label{eq:cond-exp}
    \sum_{R_i\textnormal{ is still live}}2^{-u_i}, \textnormal{ where $u_i$ is
the number of undecided chains intersecting } R_i.
\end{equation}
We can now pick any undecided Kempe chain $X$. We know that one of the two choices — swap or no-swap — will not decrease the conditional expectation, and we can determine which one it is by considering only the live rings that intersect $X$. We make a non-decreasing choice for one Kempe chain at a time. Since each ring only intersects a constant number of Kempe chains (at most 18), the total time over all Kempe chains is linear.

Since the conditional expectation is non-decreasing as we make our choices, and since we end up with all Kempe chains decided, we conclude that the final set of live rings is at least as big as the original expectation.
\end{proof}

We will now use the above lemma to decide on 25 Kempe changes. We say a ring $R_i$ is \emph{active} if all previous Kempe changes have been improving, and it is not yet 0-extendible. We apply Lemma \ref{lem:det-Kempe-change} to all active rings. We know that a fraction $p$ gets an improving move. For the rings $R_i$ whose coloring becomes 0-extendible, we extend the coloring to $Z_i$. The active rings that got an improving move but are still not 0-extendible remain active. We repeatedly apply Lemma \ref{lem:det-Kempe-change} to the active rings, finding a new Kempe change that improves by a factor of $p$ for the remaining active rings. After 25 repetitions, we know that there are no active rings left, and that at least a fraction $p^{25}$ of the original rings have been extended. Conversely, this means that the number of empty rings 
has been reduced by at least a factor $1-p^{25}$.

We now restart with all the empty rings, all of which are 25-extendible. If we started with $r$ reducible configurations, then after at most $\log_{1/(1-p^{25})} r = O(\log r) = O(\log n)$ rounds of 25 Kempe changes, we know that no empty rings remain, hence that we have colored the whole original configuration $G$.

The recursive step, reducing the problem size by a constant factor, is thus completed in $O(n\log n)$ time. This suffices for the  $O(n\log n)$ time 4-coloring in \cite{nearlinear}.

\paragraph{More efficient algorithm}
To clarify what is needed for a more efficient 4-coloring algorithm, \cite{nearlinear} stated that it reduces to the following task:
\begin{task}[Color Extension] \label{task}
Let $G$ be a plane triangulation of order $n$ with a family $\cZ$ of $\Omega(n)$ pairwise non-touching induced D-reducible configurations in the family $\cD$ from \cite{nearlinear}.

Find a fast algorithm that selects a subfamily $\cZ'\subseteq\cZ$
of size $\Omega(n)$ such that, for any 4-coloring of $G\setminus \bigcup\cZ'$, we can (again by the fast algorithm) perform $O(\log n)$ Kempe changes that result in a 4-coloring of $G\setminus \bigcup\cZ'$ which extends to all of $G$.
\end{task}

As described above, using Lemma \ref{lem:det-Kempe-change} recursively, \cite{nearlinear} solves this task with $\cZ'=\cZ$ in $O(n\log n)$ time. Moreover, \cite{nearlinear} proves:

\begin{theorem}[\cite{nearlinear}]\label{thm:recurrence}
If we can solve the Color Extension
Task \ref{task} of order $n$ in time $T(n)$, then we can 4-color any $n$-vertex planar graph in time $T^*(n)\leq O(n)+T(n)+T^*(n/2)$.
\end{theorem}

\subsection{Our contribution: linear time 4-coloring}

In this paper, we will prove the following theorem.
\begin{theorem}
\label{thm:task-linear}
    We can solve Task \ref{task} in linear time.
\end{theorem}
By Theorem \ref{thm:recurrence}, we can 4-color any planar graph in linear time, as stated in Theorem \ref{thm:linear-time}.

To solve Task \ref{task} in linear time, we need to implement the $O(\log n)$ Kempe changes in linear total time.
The idea is to spend $O(n)$ time to preprocess the graph, including the selection of an appropriate subfamily $\cZ'$, so that we can identify and simulate each Kempe change in $O(n/\log n)$ amortized time. Thus, we implement all $O(\log n)$ Kempe changes in $O(n)$ total time.

We note that creating such a representation for fast identification and simulation of good Kempe changes is a challenge of data structures, unlike \cite{nearlinear}, which focused on
a new graph theoretic understanding of reducible configurations.
Below we list the basic techniques and challenges.

\paragraph{Tabulation} Not surprisingly, our logarithmic speedup comes from tabulating answers to subproblems that can be coded with $\tfrac{1}{2}\log n$ bits, an idea that goes back to the Method of Four Russians \cite{arlazarov1970economical}. These tables are stored in arrays, viewing the codes as integer indices below $\sqrt n$.
It will be important here that a planar graph on $q$ vertices can be coded with $O(q)$ bits \cite{keeler1995short,turan1984succinct}. We will therefore be very careful in showing that various subproblems can be coded with planar graphs. The challenge is to translate the tabulated solutions for logarithmic-size problems into a full logarithmic speed-up for  4-coloring.

We note that tabulation has its hard limits in the 
sense that while space $s$ can tabulate answers 
to $s$ queries, it may be much less useful when the space of queries is larger than $s$.  A first formal proof of this is found in \cite{PT06-pred}, which considers using space 
$s=n \log^{O(1)} n$ to support predecessor queries among $n$ integers from a range $[u]$. If $u=s$, then we can answer queries in constant time with a trivial table. However, if $u=s^{1.1}$ ($1.1 \log_2 s$ bits versus $\log_2 s$ bits for keys in $[u]$), then \cite{PT06-pred} proves that the optimal query time is $\Theta(\log\log n)$ even in the powerful cell-probe model for static data structures with unlimited preprocessing time.

\paragraph{$r$-divisions}
To organize the flow of information, we shall use the recursive $r$-divisions
from \cite{Fre87,KMS13}. Here, an $r$-division of a planar graph divides it into $O(n/r)$ planar subgraphs, each of size at most $r$ and boundary of size $O(\sqrt r)$. From a tabulation perspective, $r=\Theta(\log n)$ is interesting because the whole subgraph can be coded with $\Theta(\log n)$-bits, but $r=\Theta(\log^2 n)$ is also of interest because the boundary is a planar graph that can be coded with
$\Theta(\log n)$-bits.

\paragraph{The challenge}
Our challenge is that we need to code not just the regions and their boundaries. To get a deterministic algorithm, we must collect information on all Kempe chains to determine, for each, whether swapping it maximizes the conditional expectation in \req{eq:cond-exp}. This basically implies finding out how many of the live reducible configurations want each Kempe chain swapped versus not swapped. For a region in an $r$-division, the boundary may intersect $\Theta(\sqrt r)$ Kempe chains, but we also have swap/no-swap $O(\log r)$-bit counters from all live configurations. Even if we are the most optimistic, we need $\Theta(\sqrt r\log r)$ bits of information to cross the boundary. For $r=\log^2 n/\log\log n$, this is $\Theta(\log n\sqrt {\log\log n})$, which is more than we can tabulate, yet we cannot afford to use a combinatorial algorithm since the total size of all these boundaries is $O(n/\sqrt{r})=O(n\sqrt{\log\log n}/\log n)$ and we only have $O(n/\log n)$ time available. 

The main idea to resolve this bottleneck will be to use codes based on appropriate planar graphs to control how tabulation can more efficiently transfer information between arrays of counters. The structural information coded in these planar graphs fits within the tabulation budget even when the counter arrays do not. We divide the counter arrays into smaller segments so that we with a single table look-up can process one source segment and one target segment at a time using the structural information and the previously chosen swap decisions. In the intermediate range of region sizes, the total number of these segment-pair operations is small enough to preserve the required $O(n/\log n)$ time bound. Section 3.1 gives the details.

\paragraph{The Use of Generative AI:}
The authors wrote this paper and used ChatGPT Pro and Gemini Deep Think only to assist with proofreading, including comments for the exposition.

\section{Basics}\label{sec:basics}

As our 4-coloring algorithm in \cite{nearlinear} is described so far, it is combinatorial
and takes $O(n\log n)$ time. Below we describe some of the known building blocks that we will use to get down to linear time.

\paragraph{Coding and tabulation}
Let us recall the basic terminology about coding. A \emph{(binary) code} for a set $\mathcal F$ of objects is a function $c: {\mathcal F}\to \{0,1\}^*$ (all finite binary words). We require that for $F\ne F'$, $c(F)\ne c(F')$. \emph{Encoding} of $F\in \mathcal F$ is the process of creating or using the code of $F$.

We shall use coding and tabulation in the following classic way. If we can code some subproblems using $\tfrac{1}{2}\log n$
bits on each subproblem, then we can use these bits as indices for a table/array with solutions to such subproblems. Assuming that a solution fits in an $O(\log n)$-bit word, we can look it up in constant time. In fact, this is the most important feature we use of the word RAM. 

There are at most $\sqrt n$ indices of $\tfrac{1}{2}\log n$ bits, so we can spend 
$\log^{O(1)} n$ time when we first compute the solution to a subproblem, that is, when we first populate the array with solutions. 
This means that we have polynomial time relative to the size of the subproblem, and that is plenty for the kind of subproblems we shall consider.
We shall use several different such tables, but the total space and construction time for these tables is $\tilde O(\sqrt n)$ where $\tilde O$ hides logarithmic factors.

For our analysis, we will work with some power-of-two parameter 
$\ell=\Theta(\log n)$, and aim to prove that things can be coded using $O(\ell)$ bits.
In the end, this means they can be coded with $C\ell$ bits for some constant $C$. Choosing a power-of-two parameter $\ell$, such that $(\log n)/(4C) \leq \ell \leq (\log n)/(2C)$, implies that
we can tabulate such codes.

One of the important techniques in the coding of information is that if we know an enumeration $1,\ldots,q$ of objects, e.g., vertices, then we can use an array to associate $b$ bits of information with each object using only $bq$ bits in total. 
This contrasts the situation where we have to use $O(\log q)$ bits to explain which object some information should be associated with. Here $b$ could be a constant, e.g., if the information is a color, but we shall use a super constant $b$ if the information is a counter.

Often we shall use an $O(\ell)$ bit code to code several objects each occupying a segment of the code. For these segments, we often want prefix-free coding of objects. A classic and simple way to do this is as follows. If something like the above array is encoded with $b$ bits, then we prefix this with an $O(\log b)$ bit prefix-free coding of $b$, e.g., 
a unary description of $\lceil \log_2 b\rceil$ consisting of $\lceil \log_2 b\rceil-1$
\texttt{0}s followed by a \texttt{1}, and then follow this  by the binary representation of $b$.
With this representation, we can uniquely read the objects stored in the code, one at a time. 

\paragraph{Divisions}
The idea of $r$-divisions stems from \cite{Fre87}, but we shall use the
version from the linear time construction in \cite{KMS13}. Below we review the relevant results from 
\cite{KMS13}, but using terminology more in line with other terminology used in the current paper. We will describe the translation between the definitions in Appendix \ref{app:translate}.

A \emph{region} $R$ of a triangulation $G$ is a connected subgraph
of $G$ with at least one edge. The faces of $R$ are \emph{inner} if they are faces of $G$;
otherwise they are \emph{holes}. A region has only a constant number
of holes\footnote{The number of holes is an unspecified constant from
the construction in \cite{KMS13}.}. 
Note that if a region has exactly one hole, then it is a
near-triangulation where the hole is the outer face. The vertices and edges around a hole are called \emph{boundary vertices and edges}. Their union taken over all holes of $R$ forms the \emph{boundary} $\boundary R\subseteq R$.

A \emph{division} of $G$ is a collection of regions $R_1,\ldots,R_k$
such that every face of $G$ is an inner face in exactly one region $R_i$. Moreover, it is required that if an edge is shared by two or more regions, then it has to be a boundary edge (that is, on a hole) of any region containing it. In fact, the converse is also true, for if an edge $uv$ is on the boundary of $R_i$, then one of its incident triangles $\Delta$ in $G$ is in a hole of $R_i$, and $\Delta$ has to be an inner face of another region $R_j$ which then also contains $uv$.

An \emph{$r$-division} of $G$ is a division with $O(n/r)$ regions such that every region has at most $r$ vertices and at most $O(\sqrt r)$ boundary vertices.

A \emph{recursive division} is a rooted tree of regions such that the leaves form
a division of $G$ and the region of any node is the union of its
descending leaves. It follows that the root region is $G$ itself, and that any maximal anti-chain of regions in the recursive subdivision forms a division. We will also say that
the children regions $R_1,\ldots,R_d$ of a region $R$ form a division of
$R$ but their boundaries are still defined via $G$. A face of $R_j$ is inner only if it is an inner face of $R$.
If $h$ is the depth of the tree, the vertices at depth $h-i$ ($0\le i\le h$) are said to be on \emph{level} $i$ in the tree.

Given an exponentially increasing sequence $\vec r=(r_0,\ldots,r_h)$ with $r_h\ge |G|$,
a \emph{recursive $\vec r$-division} of $G$ is a
recursive division where all leaves are at level 0 and where the
regions at level $i$ form an $r_i$-division of $G$.
Moreover, regions on level $i$ have $O(r_i/r_{i-1})$ children on level $i-1$.

\begin{theorem}[\cite{KMS13}] 
\label{thm:8}
Let $G$ be a planar graph. If $r_0\ge s$, where $s$ is a sufficiently large constant, and $(r_0,\ldots,r_h)$ is an exponentially increasing sequence with $r_h\ge |G|$, then 
we can construct a recursive $\vec r$-division of $G$ in linear time. 
\end{theorem}
Below we choose $s$ to be a sufficiently large power of $2$ and define $r_i=s 2^i$.

For any region $R$, we define an \emph{inner face} of $\boundary R$ to be a face of $\boundary R$ that is not a hole of $R$. 
All of $R\setminus\boundary R$ is thus embedded in the inner faces of $\boundary R$.

\subsection{Coding the recursive $\vec r$-division}
We will now show how to efficiently encode the above recursive subdivision in a way that
will be useful for a linear-time 4-coloring algorithm.

\paragraph{Recursively encoding a planar graph with a linear number of bits}
A plane graph $G$ of order $n$ can be trivially encoded with $O(n\log n)$ bits, using $O(\log n)$-bit labels for the vertices. It is also known \cite{turan1984succinct} that we can get down to $O(n)$ bits with an efficient coding as in \cite{keeler1995short}. 
We will now show that we can use the recursive subdivision $\cR$ of Theorem \ref{thm:8} to derive an
$O(n)$-bit representation of $G$. This encoding is fairly standard, but the way we do it will be useful for understanding later 
encodings. The encoding will also
allow us to encode $O(1)$ bits of extra information with each vertex and edge.

We assume that the recursive subdivision $\cR$ is ordered in the sense that the children of each region appear in some order. The coding will be
generated in a depth-first traversal of $\cR$ visiting children respecting the ordering. The encoding of
a region $R$, including the
recursive coding of its descendants, will form a consecutive segment of the entire code for $G$.

When we
encode a region $R$, we are given labels for 
the boundary vertices in $\boundary R$. We will then
produce a coding of $R$ prefixed by the length of the code for $R$. The encoding should include
the given labels of the boundary vertices so that they can be identified in the encoding.

Consider a level-$i$ region $R\in\cR$. It has
$O(\sqrt{r_i})$ boundary vertices with given $O(\log r_i)$-bit labels. 
The region $R$ has children
$R_1,\ldots,R_d$ where $d=O(r_i/r_{i-1})=O(1)$. Let $B=\bigcup_{j=1}^d\boundary R_j\supseteq \boundary R$ be the graph
consisting of the boundary vertices and edges of these
child regions. Then $B$ is a plane graph of size $O(\sqrt{r_i})$.  We
make a local labeling of the vertices of $B$ using labels
$1,\ldots,|V(B)|$. We want to encode the plane graph $B$ by including 
these local labels plus a mapping from the given
boundary labels in $\boundary R$ to the new labels in $B=\bigcup_{j=1}^d\boundary R_j$. We can easily do this using $O(\sqrt{r_i}\log r_i)$ bits. The encoding for the child region 
$R_j$ is given our new
labels on $B$ for its
boundary vertices in $\boundary R_j\subseteq B$.

When we decode our encoding top-down, we will already have embedded
$\boundary R$ including the given
labeling of its vertices. Since we can identify these vertices in $B$, we can add $B$ to the embedding, placing $B\setminus \boundary R$ in the inner faces of $\boundary R$. This includes our new local labeling of all vertices in $B$, which is used recursively as input for the decoding and embedding of the child regions $R_1,\ldots,R_d$. 

At level $0$, a region has constant size, and we encode it explicitly, respecting its boundary labels, using $O(1)$ bits.

The reason this is efficient is that on level $i$ in the recursive
$\vec r$-division, we only have $O(n/r_i)$ regions, so the space
used to get from level $i$ to level $i-1$ is
$O((n/r_i)\sqrt{r_i}\log r_i)=O((n (\log r_i)/\sqrt{r_i})$ bits.
The total space over all levels is thus $O(n)$ bits. 
The encoding takes linear
time given the recursive division which is constructed in $O(n)$ time in \cite{KMS13}. 

Our encoding describes both $G$ and the ordered recursive division $\cR$. From the ordering, we get an ordering of the edges; namely, the order of first appearance of edges in the terminal region encodings. 
We can also use this to order the vertices based on their first occurrence in the edge ordering. These orderings imply enumerations of the vertices and edges, so we can use arrays to associate a constant amount of information with each vertex and edge, using only $O(n)$ bits in total.

\paragraph{More efficient encoding of boundaries}
We will now improve the coding of boundaries. A
region $R$ on level $i$ has
$O(\sqrt{r_i})$ boundary vertices, including
the boundary vertices of its children. Above we encoded
it with $O(\sqrt{r_i}\log r_i)$ bits using $O(\log r_i)$ bit labels for the
boundary vertices. The main reason for the labels is 
that in order to decode the
code, we need to know where each externally labeled boundary vertex from $\boundary R$ is placed in the graph representing $R$. We will now get down to $O(\sqrt{r_i})$ bits 
exploiting that a region $R$ has only a constant number of holes.

Suppose we have a plane graph $W$ whose edges
are the edges on the boundary of a hole plus an oriented start edge $uv$, oriented following the 
clockwise walk around the hole. We will then follow
the clockwise walk around the hole if we always turn right, that is, edge $xy$ is followed by edge $yz$ 
where $z$ is the neighbor of $y$ in $W$ that follows $x$ anti-clockwise. An important
point in the description of the above
boundary walk is that given
the oriented edge $uv$, we do not need any further description of the hole, since the walk defines the hole. 

Now consider a region $R$ with
$h=O(1)$ holes with boundaries
$W_1,\ldots,W_h$. Then
$\boundary R=\bigcup_{k=1}^h W_k$. To represent the above
hole walks, we just need to
store for each edge in $\boundary R$ which boundary $W_k$ it belongs to, and whether it is the start edge $u_kv_k$
and, if so, with what orientation. We call this information the \emph{hole information} associated with
each edge of $\boundary R$, and note that it takes only a constant number of bits per edge. We will have
such hole information in mind for the boundary of every region $R$ in our recursive subdivision.
With the above information, we also have a unique enumeration of the vertices and edges in $\boundary R$: we do the walks of
$W_1,\ldots,W_h$, one at a time, and enumerate the vertices and edges as they are encountered for the first time in these walks.

Recall that our recursive issue is to identify the vertices in the boundary of a region among the boundaries of its children. We are given a region $R_0$ with $O(1)$ children $R_1,\ldots, R_d$, and our goal is to encode all their boundaries in such a way that we can identify boundary vertices of the parent region $R_0$ among the boundary vertices of the children $R_1,\ldots, R_d$. 
We shall refer to this as the \emph{boundary
structure} $\boundary (R_0;R_1,\ldots,R_d)$. 

First we encode the plane graph $B=\bigcup_{j=0}^d \boundary R_j$, which we note includes $\boundary R_0$. 
We want to apply the previous encoding of any 2-connected triangulation in linear
time using a linear number
of bits. To apply this to 
$B$ we add a linear number of auxiliary edges, so as to
create a 2-connected triangulation $B^\Delta$. We
now get the encoding of $B^\Delta$ including our ordered recursive
division, which also defines a linear
order of the edges; namely, the first time they are encountered in terminal-region encodings in the
ordered depth-first traversal. 
Using the linear order of the
edges, we can use an array to
store a constant number of bits with each edge.

For each edge in $B^\Delta$, we store if it is an auxiliary edge. If not, for $j=0,...,d$, we store if it belongs to $\boundary R_j$ and if so, its hole information in $\boundary R_j$. All this adds up to a constant number of bits per edge, as desired. We can thus store the whole boundary structure of a region on level $i$ using only $O(\sqrt{r_i})$ bits.

To clarify the parent-child relation, note that the boundary of each region $R$ will appear twice
in the above boundary structures, once as a child, and once as a parent.
In each of these boundary structures, we know the edges of $\boundary R$, including the associated hole information, and this provides a complete map between the hole walks of $R$ in the two boundary structures.

Note that the above does not give us a more efficient encoding of the planar graphs, but it does give us a tight encoding of the relation between a parent $R_0$ and its children $R_1,\ldots, R_d$. We shall sometimes use this boundary structure coding when $\bigcup_{j=0}^d \boundary R_j$ is contained inside a larger planar graph $H$.

\begin{lemma}
For a region $R_0$ with children $R_1,\ldots,R_d$, suppose we have a plane graph $H$ containing $\bigcup_{j=0}^d\partial R_j$. Then we can encode $H$ using $O(|H|)$ bits so that, for each $j=0,\ldots,d$, we can identify the hole walks of $R_j$ in $H$.
\end{lemma}

\begin{proof}
Add auxiliary edges to triangulate $H$ and make it 2-connected, and use the linear-bit encoding above. The edge order lets us store a constant number of bits with every edge. These bits say whether the edge is auxiliary, to which $\partial R_j$ it belongs, and the hole information for that $\partial R_j$. From this, after decoding $H$, we recover every $\partial R_j$ together with its enumeration. The total number of bits is $O(|H|)$.
\end{proof}

The point is that below we are going to augment $\bigcup_{j=0}^d\partial R_j$ with certain Kempe nodes and configuration nodes, and then we need to know how these connect to our boundary structure.

\paragraph{Leaf regions of logarithmic size and strictly containing our family $\cZ'$ of reducible configurations.}
We are only going to use the recursive $\vec r$-division down to some $i_0$ where $r_{i_0}=s2^{i_0}= \ell$ where $\ell=\Theta(\log n)$ is
the parameter to be chosen later such that, informally speaking, $O(\ell)<\tfrac{1}{2}\log n$.

An $i_0$ region thus has at most $\ell$
vertices, implying that we
can encode it with $O(\ell)<\tfrac{1}{2}\log n$ bits, including
a coloring of the vertices. There are
at most $\sqrt n$ such representations, so we can afford to make tables over
all these representations, e.g., we can precompute the reducible configurations inside of each such leaf-region in $O(\sqrt n\log n)=o(n)$ time.

Recall from Task \ref{task} that we are given a family $\cZ$ of $\Omega(n)$ non-touching induced D-reducible configurations, and that we have to select a subfamily $\cZ'$ of linear size. We select
\[\cZ'=\{Z\in \cZ\mid Z\textnormal{ is strictly contained in some leaf region $R$, that is, }Z\subseteq R-\boundary R\}.\]
The total number of boundary vertices of leaf regions is $O(n/\sqrt{\log n})$. Since our $\Omega(n)$ reducible configurations are non-touching and each has only a constant number of vertices, we still have $\Omega(n)$ reducible configurations for which no vertex is a boundary vertex of a leaf region. These are the reducible configurations $Z$ in $\cZ'$. If $Z$ is in the region $R$, then its ring vertices may 
intersect $\boundary R$, but its completion $\hat Z$, that is $Z$ completed with its ring, is still contained in $R$.

\paragraph{Small regions with logarithmic boundaries}
We are going to use tabulation, not just for the leaf regions, but for the boundary structure of regions
up to the level $i_1$ with $r_{i_1}=\Theta(\ell^2)$. 
As described above, we can encode these boundary
structures using $O(\sqrt {r_{i_1}})=O(\ell)$ bits. 

All the above encoding is done as a preprocessing, once and for all, in linear time.
We shall use it to identify and perform Kempe changes in $O(n/\log n)$ time.

\subsection{Identifying Kempe chains}\label{sec:ident-Kemp}
We have a graph that is 4-colored except for some uncolored reducible configurations. Since we have an enumeration of the vertices, the colors can be stored in an array with $3$ bits per vertex (we use a third bit for ``uncolored").

We now want to identify all Kempe chains in $G$ based on given opposing color
pairs $\{a,b\}$ and $\{c,d\}$ with $\{a,b,c,d\}=\{1,2,3,4\}$. We will only
spend $O(n/\log n)$ total time. Our first goal is to find out how
Kempe chains connect the vertices in the rings of the uncolored reducible
configurations.

\paragraph{Accounting}
We will get down to $O(n/\log n)$ time by spending constant time on regions at levels $i=i_0,\ldots,i_1$, that is, starting from the leaf regions. We have $O(n/\log n)$ such regions.

For $i>i_1$, the boundary size is $O(\sqrt{r_i})$, and here we will spend $O(\sqrt{r_i})$ time. For the
whole level, this is $O(\sqrt{r_i}\, n/r_i)=O(n/\sqrt{r_i})$, so
level $i=i_1+1$ dominates\footnote{Strictly speaking, the time spent at level $i_1+1$ can be smaller than the time at some of the previous levels, but the upper bound used in our estimates dominates the upper bounds used on lower levels.}
with a cost of $O(n/\log n)$.

\paragraph{Kempe connectivity via Kempe edges}
From now on, we fix complementary pairs $ab$ and $cd$ of colors.
The Kempe chains only involve the colored vertices.
Since we have a proper coloring, there are no monochromatic edges, so
the $ab$-Kempe chains are the $ab$-colored components.
With the opposing color pairs given, the 4-colored 
part of the graph is partitioned into
Kempe chains. 
We shall refer to the $ab$- and $cd$-colored edges as \emph{Kempe edges}.

Our goal is to represent which vertices are connected by Kempe chains, and call two vertices \emph{Kempe connected} if they belong to the same Kempe chain. To do this efficiently, we will
manipulate the graph, introducing new \emph{Kempe nodes} and edges, but such
that the Kempe connectivity between original vertices is not changed.
Since we only care about Kempe connectivity,
we allow ourselves to add auxiliary Kempe edges between vertices of the same
color, as long as they were already Kempe connected.

\paragraph{Contracting Kempe edges in faces}
Consider two subgraphs $B$ and $H$ of the same plane graphs, e.g., $B$
could be the boundary of a region $R$ of $G$ (and $H$ could consist of some Kempe chains). We need an operation that \emph{contracts the Kempe chains from $H$ in the faces of
$B$}. The resulting graph, denoted $B^\cK_H$, is a plane graph of size
$O(|B|)$ containing $B$ and such that any two nodes from $B$ are 
Kempe connected in $B\cup H$ if and only if they are Kempe connected
in $B^\cK_H$.
Figure \ref{fig:boundary_kempe} shows an example with $\boundary R^\cK_G$ for some region $R$ in $G$.

We consider one face $\Phi$ of $B$ at a time. First we delete all parts of $H$ in $\Phi$ that are not Kempe connected to $B$.
Next, for each appearance of a vertex $v$ 
on the boundary of $\Phi$, having some incident edges in $H$, we make a copy $v_\Phi$ inside $\Phi$,
attached to $v$ by a new Kempe edge $vv_\Phi$, 
called an ``attachment edge". This edge is attached to the boundary of $\Phi$ at the same appearance of $v$. Moreover,
we move each Kempe edge $vw\in\Phi$ so it becomes a Kempe edge
$v_{\Phi}w$.
Note that if the same vertex $v$ appears multiple times along the boundary of $\Phi$, then we make a separate copy $v'_{\Phi}$ with a separate attachment edge. Also note that if we had an edge $vw$ contained in $\Phi$ and with both end-points in the boundary, then it will become an edge $v_\Phi w_\Phi$.
When we are done, we have moved all original Kempe edges in $\Phi$ so that they are now embedded strictly inside $\Phi$. Finally, we contract all Kempe edges strictly inside $\Phi$, calling the result \emph{Kempe nodes}. The new attachment edges now connect each boundary vertex of $\Phi$
with a Kempe node. The Kempe nodes and the attachment edges now
provide exactly the same Kempe connections between boundary vertices of $\Phi$ (all in $B$)
as did the Kempe edges from $H$ in $\Phi$. Doing this operation
on all faces of $B$, we get the graph $B^\cK_H$.

Finally we argue that $|B^\cK_H|=O(|B|)$. The point is that for each
vertex $v\in B$, we add at most one $v_\Phi$ each time a boundary walk of
$\Phi$ visits $v$. The number of such vertices added in $\Phi$ is at most the total length of the boundary walks of $\Phi$. Each edge bounds at most two
faces, so the number of vertices added is at most twice the number of
edges, which, in turn, is at most thrice the number of vertices in $B$.
Thus we have proved

\begin{lemma}\label{lem:face-contract} Let $B^\cK_H$ be the contraction of the Kempe
edges from $H$ in the faces of $B$. Then $B^\cK_H$
is a plane graph of size $O(|B|)$ such that any two vertices in
$B$ are Kempe connected in $B \cup H$ if and only if they are Kempe
connected in $B^\cK_H$.
\end{lemma}

\begin{figure}[htbp]
    \centering
    \begin{minipage}[b]{0.48\columnwidth}
        \centering
        \includegraphics[width=0.98\linewidth]{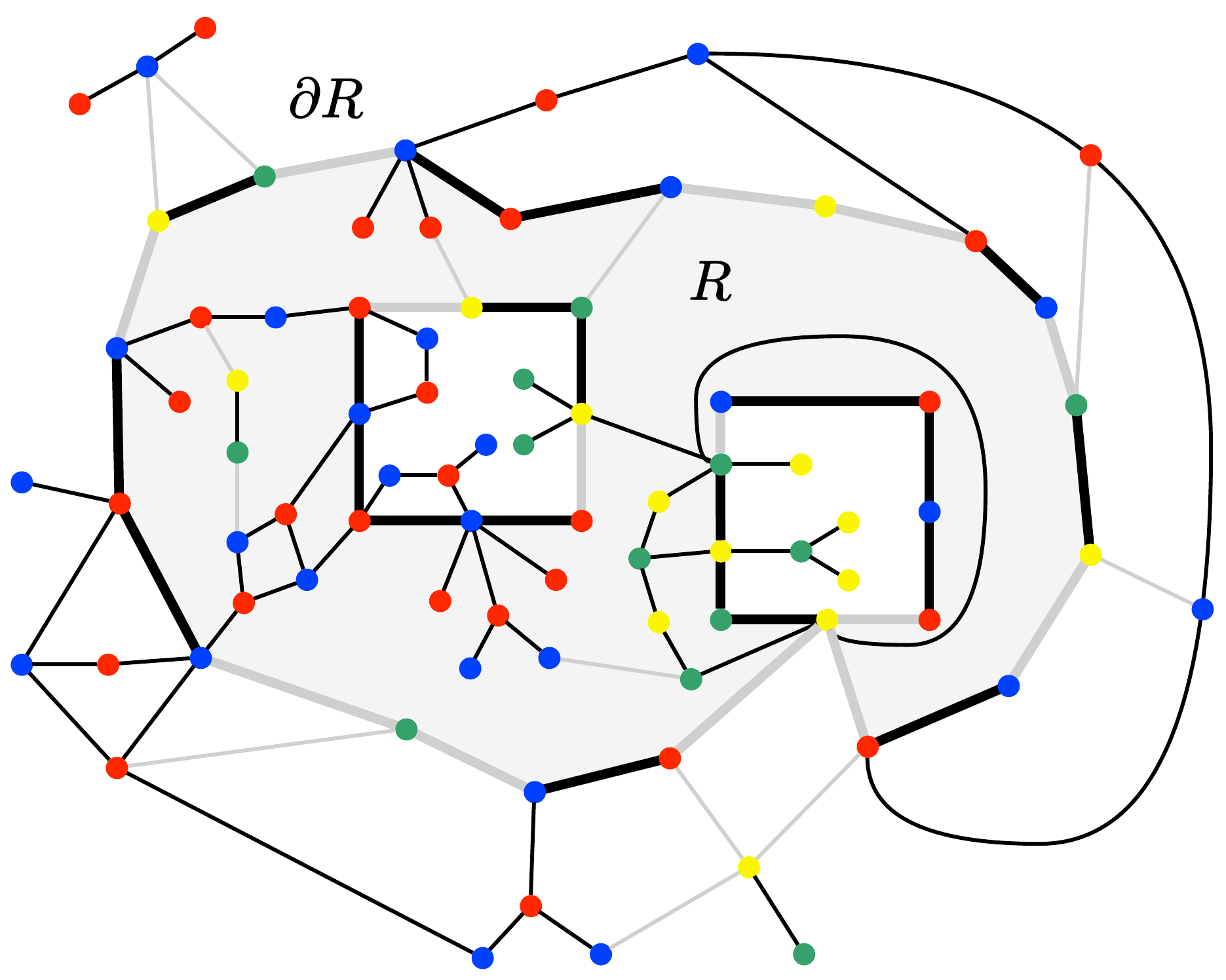}
        \caption{The figure shows  parts of a 4-colored graph $G$ with a region $R$ (light gray background) and its boundary $\boundary R$ (bold edges). 
        The black edges are those in the Kempe chains with the color pairing {\color{red}red}-{\color{blue} blue} and {\color{forestgreen}green}-{\color{bananayellow}yellow}. Other edges are drawn in gray. Note that $G$ is supposed to be triangulated, so the drawing only includes some illustrative parts of $G$. 
Also, the illustration does not include reducible configurations that are not colored (they  do not intersect any boundaries).}
        \label{fig:original_kempe}
    \end{minipage}
    \quad
    \begin{minipage}[b]{0.48\columnwidth}
        \centering
        \includegraphics[width=0.98\linewidth]{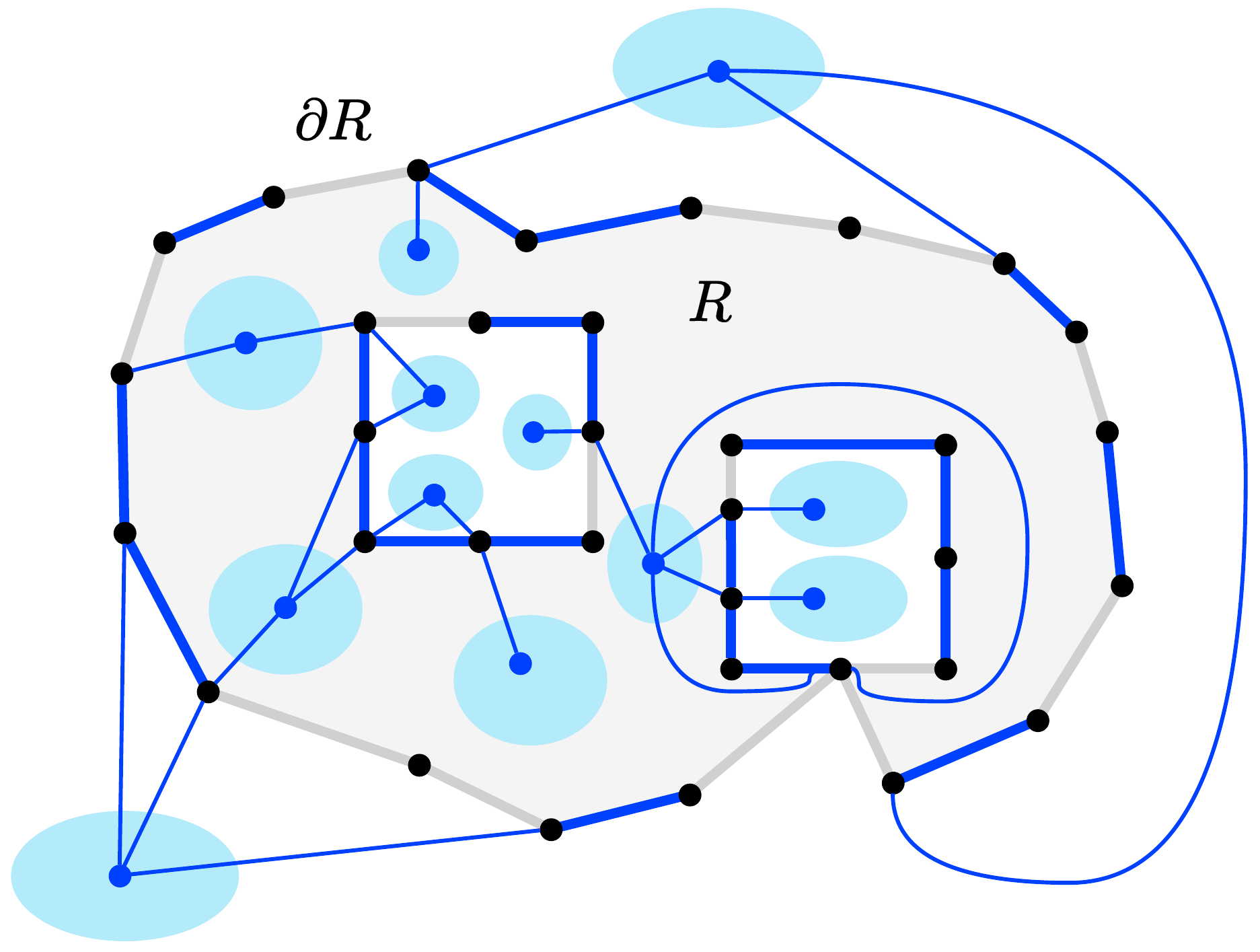}
        \caption{The figure shows the graph ${\boundary R}^\cK_G$ that we get from the Kempe chains in the left figure. The blue nodes are Kempe nodes and the blue edges are Kempe edges. They represent Kempe connectivity between vertices in $\boundary R$.}
        \label{fig:boundary_kempe}
    \end{minipage}
\end{figure}

\begin{figure}[htbp]
    \centering
    \includegraphics[width=0.44\columnwidth]{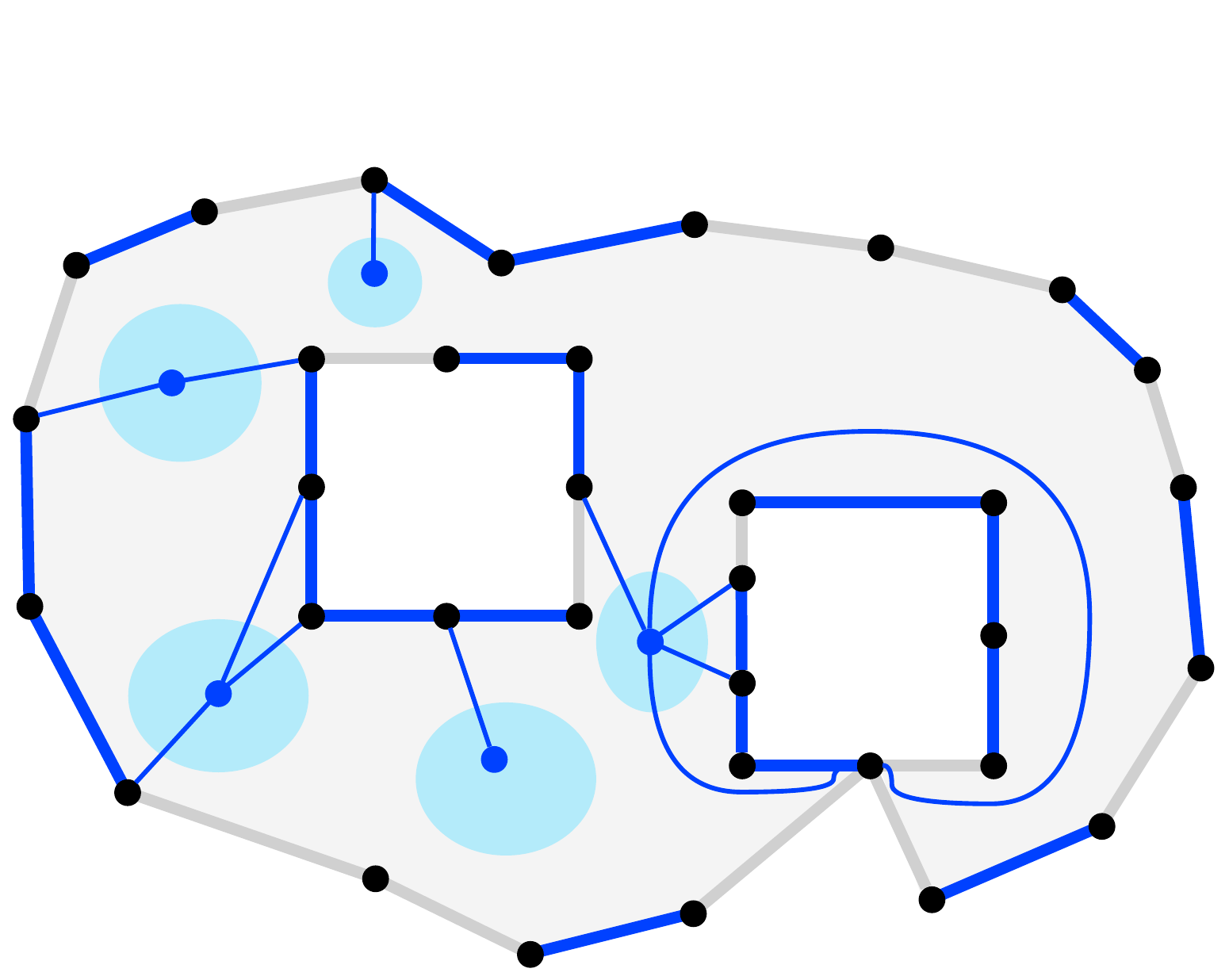}
    \includegraphics[width=0.48\columnwidth]{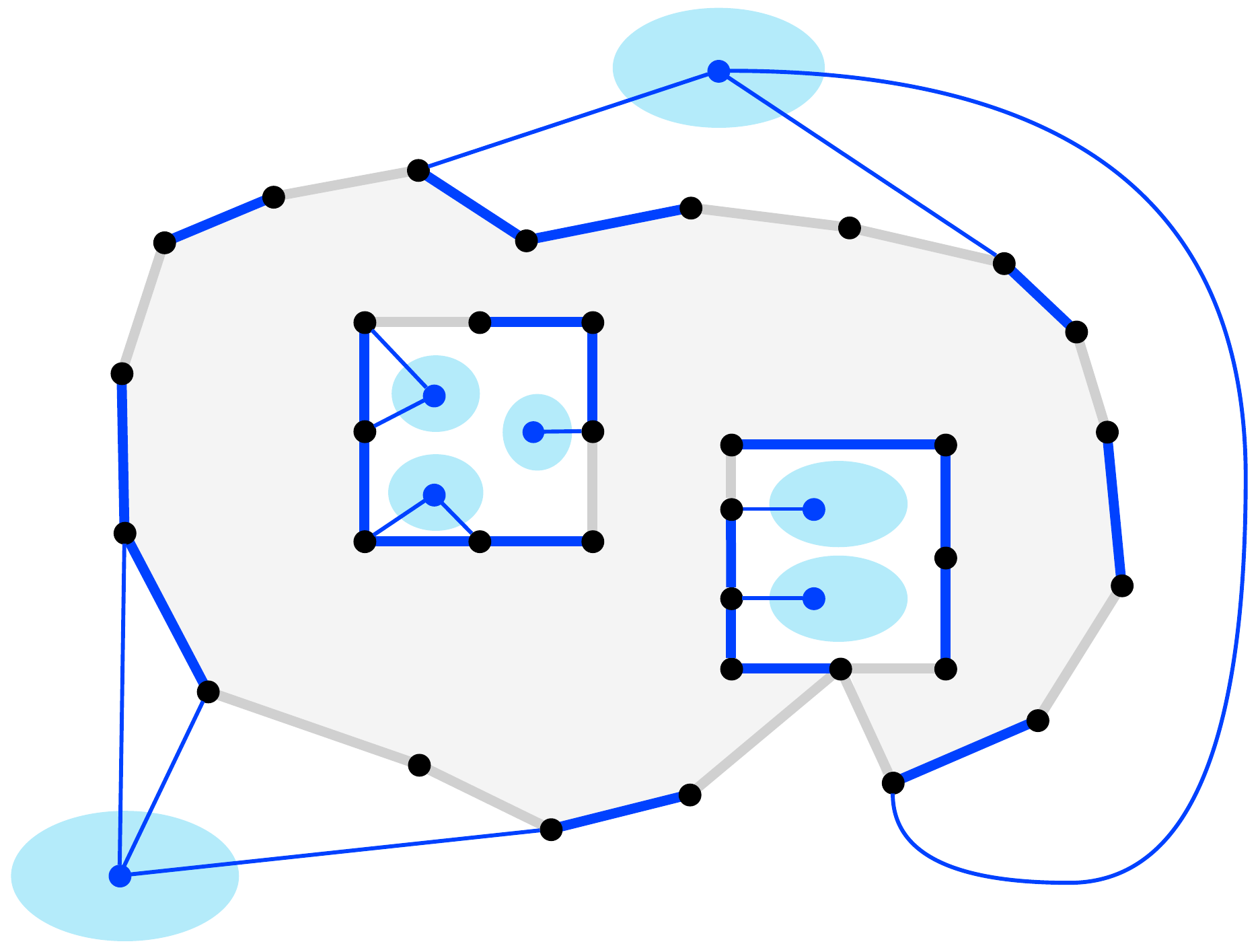}
    \caption{We produce ${\boundary R}^\cK_R$ (left) bottom-up, and then ${\boundary R}^\cK_{G-E(R)}$ (right) top-down, and finally combine them to produce ${\boundary R}^\cK_G$ (Figure \ref{fig:boundary_kempe}).}
    \label{fig:R_kempe,outside_kempe}
\end{figure}
  
\paragraph{Final goal}
For every region $R$, we will compute both ${\boundary  R}^\cK_R$
and ${\boundary  R}^\cK_{G-E(R)}$ (see Figure \ref{fig:boundary_kempe}). Note that ${\boundary  R}^\cK_{G-E(R)}$ excludes all vertices strictly contained in $R$ since they are not (Kempe) connected to $\boundary R$ in ${\boundary  R} \cup (G-E(R))$.  We also note that
${\boundary  R}^\cK_G={\boundary  R}^\cK_R \cup {\boundary  R}^\cK_{G-E(R)}$, but we
compute ${\boundary  R}^\cK_R$
and ${\boundary  R}^\cK_{G-E(R)}$ separately. We compute ${\boundary  R}^\cK_R$
bottom-up, and using them, we will compute the ${\boundary  R}^\cK_{G-E(R)}$
top-down. Note that for a leaf region $R$, we can use $R\cup {\boundary  R}^\cK_{G-E(R)}$ to find out which vertices in $R$ are Kempe connected in $G$. In particular, if we have a reducible configuration $Z\in \cZ'$ 
strictly contained in $R$, then we can see which of its ring vertices are Kempe connected, and since $R\cup {\boundary  R}^\cK_{G-E(R)}$ has size $O(\ell)$, this can be done simultaneously for all such configurations in $R$ with a single table look-up.

\paragraph{Bottom-up processing} We now first do a bottom-up processing
of our recursion tree, starting from the leaf regions. For
each region, we wish to compute ${\boundary R}^\cK_{R}$.
Given the encoding of a leaf region $R$ and its coloring, we can compute
${\boundary R}^\cK_{R}$ with a table look-up.

Above the leaf-level, consider a region $R$ with $O(1)$ child regions
$R_1,\ldots,R_d$. From
each child, we get the (encoding of the) graph ${\boundary R_j}^\cK_{R_j}$. Let
$H$ be the union of these graphs. Then
\begin{equation}\label{eq:bottom-up}
{\boundary R}^\cK_{R} = {\boundary R}^\cK_{H}\textnormal{~ where ~}H = \cup_{j=1}^d {\boundary R_j}^\cK_{R_j} = \Big(\cup_{j=1}^d \boundary R_j\Big)^\cK_R\textnormal~.
\end{equation}
The basic point here
is that if $R$ is of level $i$, then  the above only involves $d+1=O(1)$ boundaries of order $O(\sqrt{r_i})$. 
For small regions $R$ of levels $i\leq i_1$, the size is at most $\ell^2$, but the boundary size $|\boundary R|$ is only $O(\ell)$, so the computations can be done by table look-ups.

As we go regions $R$ on levels $i>i_1$, we compute
${\boundary R}^\cK_{R}$ in time linear in the total size of the graphs involved in the bottom-up step \req{eq:bottom-up}, which is $O(\sqrt{r_i})$. For the  whole level $i$, this is $O(n/\sqrt {r_i})=
O(n/\sqrt {2^i})$ time.
Since this is only for $i>i_1$ where $2^i=\Omega(\log^2 n)$, the total cost
of the bottom-up processing is $O(n/\log n)$.

\paragraph{Top-down processing}
Having computed ${\boundary R}^\cK_{R}$ for all regions $R$ on or above leaf-level $i_0$, we now wish to compute all $\boundary R^\cK_{G-E(R)}$
in a top-down processing. The top-most region is $G$, which has no boundary,
so there is nothing to compute.

Suppose we have computed $\boundary R^\cK_{G-E(R)}$ for a
region $R$ with $O(1)$ child regions $R_1,\ldots,R_d$. We compute
$(\boundary {R_j})^\cK_{G-E(R_j)}$ as follows. Let $B=\boundary R\cup\;\bigcup_{j=1}^d \boundary R_j$. Then
\begin{equation}\label{eq:top-down}
\boundary {R_j}^\cK_{G-E(R_j)}=\boundary {R_j}^\cK_{B^\cK_{G-E(R_j)}}\textnormal{ where }B^\cK_{G-E(R_j)}=B^\cK_{\boundary R^\cK_{G-E(R)}\cup\;
  \bigcup_{h\neq j} {\boundary R_h}^\cK_{R_h}}
\end{equation}
As for the bottom-up processing, the top-down processing is easily
implemented in $O(n/\log n)$ total time. As we described the final
goal, when we for each leaf region $R$ have its coloring
and $\boundary {R}^\cK_{G-E(R)}$, then we know which of
its vertices are in the same Kempe chains, and since this
is all encoded in $O(\ell)$-bits, we can get the information with
table look-ups.

\paragraph{Enumeration of Kempe chains}

Consider a non-leaf region $R_0$ with child regions $R_1,\ldots,R_d$ and let $B=\bigcup_{j=1}^d \boundary R_j$, noting that  $\boundary R_0\subseteq B$.
In our bottom-up processing, we computed the ${\boundary R_j}^\cK_{R_j}$  for $j=1,\dots,d$, and afterwards, top-down we computed ${\boundary R_0}^\cK_{G-E(R_0)}$. Putting
them together, we get
\begin{equation}\label{eq:Kempe-boundary}
B^\cK_G={\boundary R_0}^\cK_{G-E(R_0)}\cup\;\bigcup_{j=1}^d {\boundary R_j}^\cK_{R_j}.   
\end{equation}
With this information, which only takes $O(\sqrt{r_i})$ bits, we know
precisely where different Kempe chains intersect $B$. We divide the Kempe chains
intersecting $B$ into two types. Those who intersect $\boundary R_0$
and those who don't. The latter type are said to \emph{stop} in $R_0$.
Equivalently, a Kempe chain stops in $R_0$ if and only if it 
is strictly contained in $R_0$ (that is, contained in $R_0-\boundary R_0$), but not strictly contained
in any child of $R_0$.

Generally, if we have an enumeration order on a set of vertices $U$, then it induces an enumeration order of the Kempe chains intersecting $U$; namely, the order in which they are first encountered.

Recall that our boundary structure $\boundary(R_0;R_1,\ldots,R_d)$ provides unique enumeration of the vertices in $B$, as well as an enumeration of the vertices in each $\boundary R_j$. This then induces orderings of the intersecting Kempe chains, and all this information is implicit when we combine 
$\boundary(R_0;R_1,\ldots,R_d)$ with $B^\cK_G$. This is all represented in $O(\sqrt{r_i})$ bits if $R_0$ is on level $i\geq i_0$.

Combining this information with the corresponding information from the parent of $R_0$, we can follow
how the same Kempe chains appear in the different enumeration orders.

For a leaf region $R$, we identify all intersecting Kempe chains from
$\boundary {R}^\cK_{G-E(R)}$
together with $R$, all represented in $O(\ell)$ bits. The Kempe chains are enumerated as above. The Kempe chains strictly contained in a leaf region $R$ are said to stop in $R$. Then every Kempe chain in $G$ stops in exactly one region.

\paragraph{Implementing Kempe changes}
For every region $R$ in our recursive subdivision, we have an enumeration
of the Kempe chains stopping in $R$. Suppose
we are given a bit-string telling for each of these Kempe chains
if they should be swapped or not, e.g., these bits could be chosen at
random for a random Kempe change. We will now show how to implement these given changes in the coloring of $G$.

The implementation of the Kempe changes is done top-down. Consider a region
$R$. We had an enumeration of the Kempe chains intersecting $\boundary R$, and we assume that we, from the above, are given a bit string for
the decisions for these Kempe chains. Now we add to this the decisions
for the Kempe chains stopping in $R$. If $R$ had children $R_1,\ldots,R_d$,
then for each child $R_j$, we need to translate from these bit-strings to
a bit-string for the Kempe chains intersecting $\boundary R_j$ following
their enumeration for $R_j$. However, we know the translation from  $\boundary(R_0;R_1,\ldots,R_d)$ and $B^\cK_G$ where $B=\bigcup_{j=1}^d\boundary R_j$.

If $R$ is on level $i>i_1$, then this involves $O(\sqrt{r_i})$ Kempe chains, and
all the information can be described using $O(\sqrt{r_i})$ bits.
On levels $i$ such that $i_0<i\leq i_1$, we have $\sqrt{r_i}=O(\ell)$, so the translation can be done with a table look-up.  On the high level,
for $i>i_1$, the translation is trivially done in $O(\sqrt{r_i})$ time.

When we get to a leaf region $R$, from above, we get the bit string
deciding the changes for Kempe chains intersecting $\boundary R$, and we combine this with the bit strings for the Kempe chains strictly contained in
$R$. We now know exactly which vertices in $R$ whose colors we want to swap. Boundary vertices are in multiple leaf regions. The color array for each leaf stores a local copy of every vertex color in that leaf.  During a Kempe change, all these local copies are updated. Copies of the same vertex receive the same Kempe chain decision, so their colors remain equal. They are initialized consistently during preprocessing, and their total number is $O(n)$. All these colors are found in the color array for vertices in $R$. Using a single table look-up 
based on $R$, $\boundary {R}^\cK_{G-E(R)}$, the Kempe chain swap vector, and the color array, we produce an updated color array that replaces the old one. All of these objects are represented in $O(\ell)$ bits, so the color copies in the leaf region are all updated in constant time.

We conclude that we can identify Kempe connections and implement changes in $O(n/\log n)$ time.

\subsection{4-coloring in linear expected time}
We are now ready to implement the randomized approach from Section \ref{sec:near-linear}.
Recall that we first fixed a random color pairing, and then made a random choice whether or not to switch each Kempe chain.

For each region $R$, we just need to generate a random bit string specifying the random choices for all Kempe chains stopping in $R$. The total number of Kempe chains is at most $n$. Assuming that we can generate a random word with $\log n$ bits in constant time, we generate all the bits in $O(n/\log n)$ time. Having specified the random Kempe change, we can implement it as described above in $O(n/\log n)$ time.

Next, we need to check which of the uncolored reducible configurations $Z$ from $\cZ'$ have become 0-extendible so that $Z$ can now be colored. Since they are all strictly contained in leaf regions, we can just check one leaf region $R$ at a time. It is represented with $O(\ell)$ bits, including the new coloring, so with a simple look-up, we perform this task for all
the relevant configurations strictly contained in $R$ in constant time.

As described in Section \ref{sec:near-linear}, we are done with high probability after $O(\log n)$ iterations (and one iteration (Kempe change) takes $O(n/\log n)$ time), thus spending $O(n)$ time in total. 

\section{Deciding and performing Kempe changes deterministically}
In this section, we shall derandomize the previous construction, so as to prove our main result; namely that we can 4-color deterministically in linear time. Technically, doing this derandomization efficiently is by far the most challenging part of this paper. 
The specific challenge is to prove the lemma below, 
which improves construction time of Lemma \ref{lem:det-Kempe-change} from $O(n)$ to $O(n/\log n)$:

\begin{lemma}\label{lem:fast-det-Kempe-change}
    Let $Z_1,\ldots, Z_q$ be some (not necessarily all) non-colored reducible configurations (from the family $\mathcal Z'$ selected in Section~2.1) in a triangulation $G$ that is 4-colored except for the remaining non-colored configurations in $\mathcal Z'$. The colorings of their rings 
    are not yet $0$-extendible. Suppose also that we have computed the recursive division with all the boundary structures $\boundary(R_0;R_1,\ldots,R_d)$ from above (this took linear time, but was only done once in the preprocessing).
    Finally, we assume that the participating configurations $Z_1,\ldots,Z_q$ are marked in the coded leaf representations, so that we do not need to handle them individually.

    We can then, deterministically in $O(n/\log n)$ time, find and execute a Kempe change that is improving for at least a fraction $p=\frac1{3\cdot 2^{18}}$ of the rings, and also tell which rings got  improved.
\end{lemma}

With Lemma \ref{lem:fast-det-Kempe-change}, we can solve Task \ref{task} in linear time as stated in Theorem \ref{thm:task-linear}. By Theorem \ref{thm:recurrence}, this implies that we can 4-color in linear time, so Lemma \ref{lem:fast-det-Kempe-change} implies Theorem \ref{thm:linear-time}.

The algorithmic proof of Lemma \ref{lem:fast-det-Kempe-change} follows the same steps as Lemma \ref{lem:det-Kempe-change}, but each step must be done logarithmically faster. Below we review the steps.

Before we decide on the Kempe change, we view all the rings as alive. Each time we make a decision about the Kempe change, we kill all rings that can no longer get an improving Kempe change.

First, we need to find a good color pairing: should color 1 be paired with 2, 3, or 4? This decides all the bicolored Kempe chains. Since every configuration is D-reducible, we know, for every ring, that one of these 3 pairings has an improving Kempe change. We choose a pairing such that at least 1/3 of the rings have an improving Kempe change, and thus remain alive.

With a color pairing $a, b$ and $c, d$ decided, we go through the $ab$- and $cd$-colored Kempe chain in an arbitrary order and decide for each whether or not to swap its colors. We always pick the choice that maximizes the conditional expected number of surviving rings from \req{eq:cond-exp}, that is,
\[\sum_{Z_i\textnormal{ is still live}}2^{-u_i}, \textnormal{ where $u_i$ is
the number of  undecided Kempe chains touching } Z_i.\]
When deciding on a Kempe change $K$, we only need the sum restricted to live configurations touching $K$. Deciding $K$ may kill many of these configurations, but for those $Z_i$ that are touched by $K$ but not killed, the survival probability $2^{-u_i}$ is doubled. 
Since we always have $u_i\leq\maxringsize$, we  note that the expectancy is always an integer multiple of $1/2^{\maxringsize}$.
Let us describe each step in more detail. 

\paragraph{Finding the right color pairing}
For any given color pairing, we want to know for each
non-colored configuration $Z$ if there exists a Kempe change that is improving for $Z$; for then $Z$ survives this color pairing.  We want to count the number of survivors with each of the three color pairs, so we can pick one where 1/3 survives.

Recall from the definition of $\cZ'$ that $Z$ is strictly contained in some leaf region $R$, and all we need to know is how the Kempe chains partition the ring of $Z$. From the local perspective of $Z$, a Kempe change tells whether or not to swap the colors within each set of this partition, and for each of the constantly many configurations $Z$ in $\cD$, we know which swaps are improving.

To find out how Kempe chains connect vertices in $R$, we 
just have to consider $R\cup \boundary {R}^\cK_{G-E(R)}$, where $\boundary {R}^\cK_{G-E(R)}$ is as computed 
as in Section \ref{sec:ident-Kemp}. We note that
$R\cup \boundary {R}^\cK_{G-E(R)}$ is represented in 
$O(|R|)$ bits, including coloring, reducible configurations (alive before the color pairing is decided), and all Kempe connections between its vertices. We can therefore, with a single look-up, in constant time, decide how many configurations survive the given color pairing.

Thus, for each color pairing, we first identify the Kempe chains, then count the survivors for that pairing, and choose the one with the most survivors.

With the winning pairing, for each surviving configuration, we have some preferred improving swaps among the Kempe chains it touches. We keep these preferred swaps fixed throughout the following bottom-up decisions.

\paragraph{Bottom-up swap decisions on Kempe chains}
We are going to make the swap decisions on the Kempe chains in a bottom-up fashion, starting from the leaf regions. When we get to a region $R$, we assume that we have already made decisions for all Kempe chains stopping below $R$, so when we get to $R$, we only have to make decisions
for Kempe chains $K$ stopping in $R$. However, to make the right decision on $K$, we need to know about all the live configurations touching $K$.

It is now important to note that we only consider reducible configurations $Z$ strictly contained in leaf regions, so they cannot intersect the boundary of any region. Therefore, 
$K$ can only touch configurations strictly contained in $R$.

When the decisions have been made for all the Kempe chains in the graph, we can implement them all top-down as described at the end of Section \ref{sec:ident-Kemp}.

\paragraph{Kempe chains strictly inside leaf regions}
First we consider the easy case of a leaf region $R$. A Kempe chain $K$ stops in $R$ if and only if $K$ is strictly contained in $R$.

Since we only
consider reducible configurations strictly within regions, $K$ can only
touch reducible configurations strictly within $R$. All the relevant information is therefore contained in $R\cup \boundary {R}^\cK_{G-E(R)}$, including prior decisions on Kempe chains in $R$.
In principle, we want to decide one such Kempe change $K$ at a time, but we can simulate this sequential process for all of them, just using the information from $R\cup \boundary {R}^\cK_{G-E(R)}$. Using a table, taking $R\cup \boundary {R}^\cK_{G-E(R)}$ and the original coloring of $R$ as input, we can produce a bit vector of decisions.

However, for each of the at most $|\boundary R|$ Kempe chains intersecting 
$\boundary R$, we need information about the live reducible configurations that they touch. The representation of such information is described below.

\paragraph{Deciding on batches of Kempe chains}
We will now discuss how to
decide the changes for a batch $\vec K=(K_1,\ldots,K_k)$ of Kempe chains, simulating the decided changes in this order. These
could be the Kempe chains stopping in a given region. 
The decisions depend on
all the reducible configurations touched by these Kempe chains, but there could be $\Omega(n)$ of these even for a single Kempe chain. We want to argue that
the information needed for
the decisions can be contained in $O(k)$ counters.

We need to consider reducible configurations 
$Z$ that touch some Kempe chains from $\vec K$ and which are alive before we make switch decisions on $\vec K$.  
Relative to our Kempe batch $\vec K$, the
\emph{main type} of such a reducible configuration $Z$ is the set of Kempe chains from $\vec K$ that $Z$ touches. We claim that the number of main
types relative to $\vec K$
is only $O(k)$. To see this, we construct a 
plane graph representation $\vec K^\cC$ as follows.
We contract each Kempe chain
from $\vec K$ to a Kempe node. Moreover,
we contract every configuration touched by the Kempe nodes to 
a configuration node. We remove the rest of the graph except for 
the edges between
the Kempe nodes and the configuration nodes, so now we have a plane bipartite graph representing all connections between Kempe chains from $\vec K$ and the configurations they touch. Next we remove all parallel edges since we do not care how many times a given Kempe chain touches  the same configuration. Finally, and this is the big reduction, we remove all but one
configuration node of each main type.
 
We claim that the number of remaining configuration nodes is $O(k)$.
With black Kempe nodes and white configuration nodes, this follows from the following general lemma:
\begin{lemma}\label{lem:Kempe-config}
  Consider a simple plane bipartite graph where the two partite sets are black and white.
  Suppose that every white vertex has at least one black neighbor and no two white vertices are adjacent to the same set of black vertices
  (think of this as a simple plane hypergraph where the white vertices define the hyperedges). Then the number of white vertices is within a constant factor of the number of black vertices.
\end{lemma}
\begin{proof}
Let $q$ be the number of black vertices.  If $q\leq2$, the statement is immediate since there are at most $2^q$ different sets of black neighbors. Assume $q\geq3$. There are at most $q$ white vertices of degree 1. For the white vertices of degree 2, suppress each such white vertex to an edge between its two black neighbors. Since no two white vertices have the same set of black neighbors, this gives a simple plane graph on the $q$ black vertices, and hence there are at most $3q-6$ white vertices of degree 2.

Let $u$ be the number of white vertices of degree at least 3. Delete all other white vertices. The remaining graph is a simple bipartite plane graph with $q+u$ vertices and at least $3u$ edges. Hence by Euler's formula for triangle-free planar graphs, 
\[
 3u\leq 2(q+u)-4,
\]
so $u\leq2q-4$. Thus there are at most $q+3q-6+2q-4=6q-10$ nonempty sets of black neighbors. 
\end{proof}
There can be a linear number of configurations of each main type, but we claim that we can represent all information relevant to our decisions using only a constant number of counters for each main type.

More precisely, we divide each main type into a constant number of \emph{subtypes}. If
a configuration $Z$ touches $s$ different
Kempe chains from $\vec K$, then the
subtype tells which of these should be swapped to get the preferred improving Kempe changes for $Z$. There
are $2^s$ options.

The subtype also tells the number $t$ of undecided 
Kempe chains outside $\vec K$ that $Z$ touches. Then
$s+t$ is bounded by the ring size, which is at most $\maxringsize$. Therefore, the number of subtypes is at most $2^{\maxringsize+1}$.

Each configuration node in $\vec K^\cC$ represents configurations of the same main type and has counters for how many configurations each subtype represents. Recall that we only count configurations that were alive just before the swap decisions on $\vec K$.

With the above 
representation in place, we will now decide on the changes for our Kempe chains $K_1,\ldots,K_k$.
Before we decide on $K_j$, the expected number of surviving configurations touched by $K_j$ is
a sum over all types (main types with subtypes). To contribute, the main type should touch $K_j$ and for the subtype, the preferred swap decisions on $K_1,\ldots,K_{j-1}$ should agree with those we made. If the main type touches $u$ Kempe chains in $K_j,\ldots,K_k$, then its survival probability is $1/2^{u+t}$, so the contribution is the counter times $1/2^{u+t}$.

When we decide on $K_j$, we double the above probability for subtypes agreeing with the decision, and zero the probability of those that do not agree. We make a decision that maximizes the resulting total probability.

For a more efficient implementation, we can, for each subtype, maintain whether it is still alive after the decisions made so far on $\vec K$, and if so, maintain the above number $u+t$ of undecided Kempe chains it touches.
Then, when we make the decision on $K_j$, we only involve the configuration types touching $K_j$. Since the number of configuration types is $O(k)$ and each only touches $O(1)$ Kempe chains, the total time needed to decide on $\vec K$ is $O(k)$ if we use a simple combinatorial algorithm.

Also, if we can represent $\vec K^\cC$
with all its counters in $O(\ell)$ space, then we can make all the decisions with a table look-up in constant time. However, there will 
be some situations where the combinatorial algorithm is too slow and where we have too much information for a table look-up.

\paragraph{Bottom-up construction of relevant main configuration types}
We will now show how to efficiently compute main configuration types for regions. For now we do not include subtypes and counters, and we will ignore if a main type represents a live configuration. Relative to a batch $\vec K$ of Kempe chains, the main type of a configuration is the subset of $\vec K$ that it touches. Recall that our description of $\vec K^\cC$ contracted each Kempe chain to a single Kempe node and a configuration of each main type to a configuration node. In Lemma \ref{lem:Kempe-config}, this is used
to show that the number of main types is linear in the size of $\vec K$.  However,
for the representation of Kempe chains intersecting the boundary of $R$, instead of contracting these Kempe chains completely, we will
use the plane graph $\boundary R^\cK_G$ where the
Kempe chains in $G$ intersecting $\boundary R$ are only contracted inside the faces of $\boundary R$.

We will augment the graph $\boundary R^\cK_G$ with configuration nodes from the main types
from $R$ relative to the Kempe chains intersecting $\boundary R$. More precisely, we construct $\boundary R^\cKC_{GR}$ as follows. From $R$, we keep an arbitrary 
configuration
from each main type relative to the Kempe chains intersecting $\boundary R$. No other configurations from $R$ are relevant. We contract each main-type configuration as we constructed $\vec K^\cC$. If one of these configurations has more than one edge to the same Kempe chain, we delete all but one of them. Finally, we do the contractions in the faces of $\boundary R$ of the Kempe chains intersecting $\boundary R$. What remains is $\boundary R$ with Kempe nodes and main type configuration nodes in the faces.

The number of Kempe chains intersecting $\boundary R$ is bounded by the size of $\boundary R$, and by Lemma \ref{lem:Kempe-config},  their number also bounds the number of main configuration types since this number would be the same if we completely contracted all the Kempe chains. The number of main configuration types is thus $O(|\boundary R|)$ and we already know that $\boundary R^\cK_{G}$ has size $O(|\boundary R|)$, so we conclude that $\boundary R^\cKC_{GR}$ has size $O(|\boundary R|)$.

Our representation of 
$\boundary R^\cK_{G}$ came with an enumeration of the Kempe chains
intersecting $\boundary R$,
and in $\boundary R^\cKC_{GR}$, we can see the Kempe chains touched by a configuration node.
We now get a natural enumeration of the main configuration types in $\boundary R^\cKC_{GR}$.
For each main type of a configuration $Z$, we first sort the indices of the Kempe chains touched by $Z$. Second, we order the main types lexicographically with respect to these sorted lists.
Thus we get an enumeration of all main configuration types in $\boundary R^\cKC_{GR}$.

\begin{figure}[htbp]
    \centering
    \includegraphics[width=0.55\linewidth]{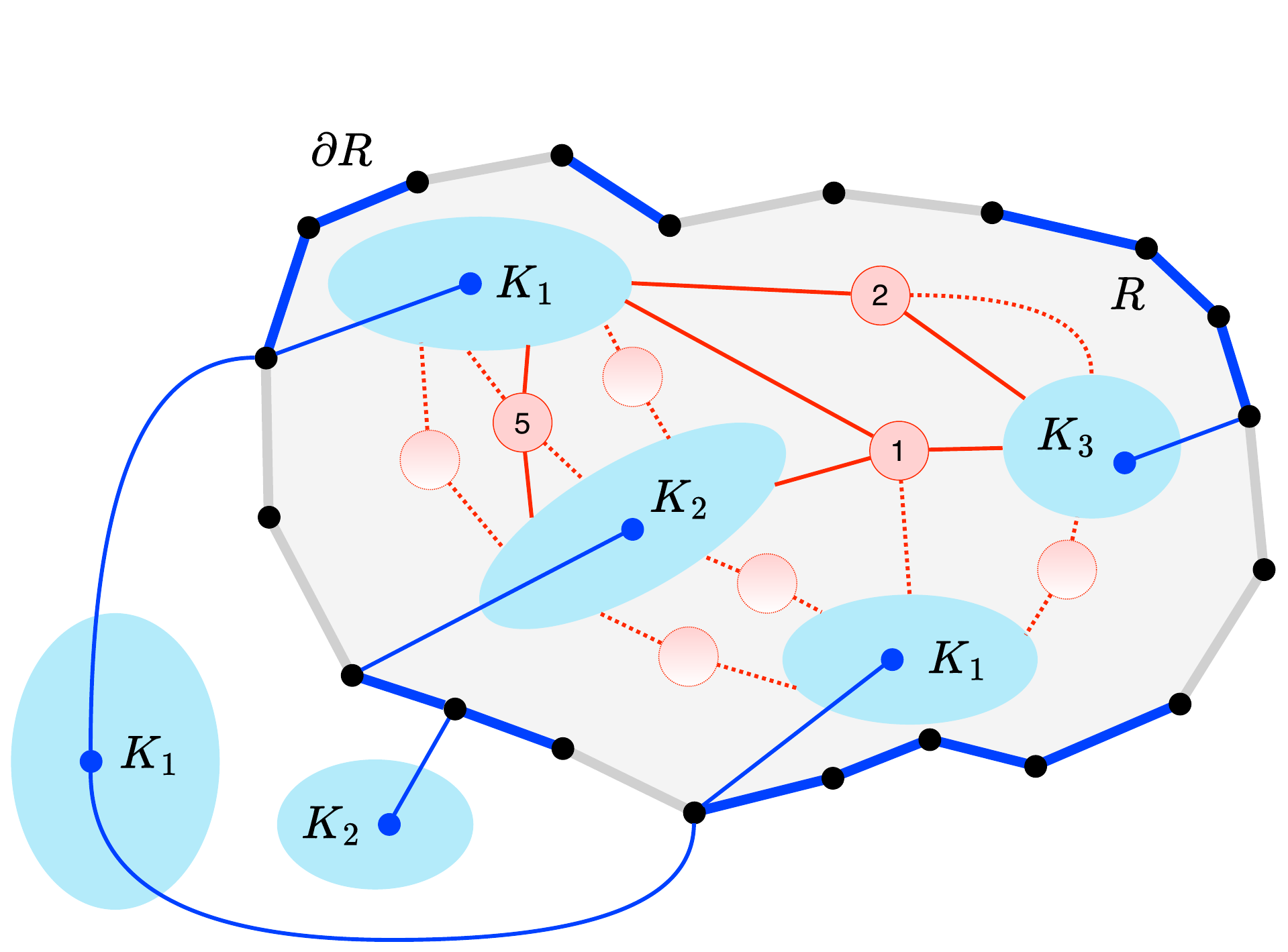}
    \caption{The figure of $\boundary R^{\cK\cC}_{GR}$. There are three main configuration types: $\{K_1,K_2\}$, $\{K_1,K_2,K_3\}$, $\{K_1,K_3\}$, and each counter is $5,1$ or $2$. For simplicity, subtypes are ignored in these counts.}
    \label{fig:boundary_kempe_conf}
\end{figure}

Now we consider
a region $R_0$ with child regions $R_1,\ldots, R_d$. Our goal
is to construct $\boundary {R_0}^\cKC_{GR_0}$ given $\boundary {R_j}^\cKC_{GR_j}$ for each child $R_j$. Let $B=\bigcup_{j=1}^d
\boundary R_j$, and recall that $\boundary R_0\subseteq B$. From  \req{eq:Kempe-boundary} we have 
\[B^\cK_G=
{\boundary R_0}^\cK_{G-E(R_0)}\cup\;\bigcup_{j=1}^d {\boundary R_j}^\cK_{R_j}. \]
When we combine $B^\cK_G$ with the boundary structure $\boundary (R_0;R_1,\ldots,R_d)$, we have all the information needed to follow how the Kempe chains intersect the boundaries of the holes of each individual boundary $\boundary R_j$, including the parent boundary $\boundary R_0$.

We now want to construct 
$B^\cKC_{GR}$ which
is $B^\cK_G$ with a configuration node from each
main type in $R_0$ added. We will
construct it from the 
${\boundary R_j}^\cKC_{GR_j}$
that we get from the child regions $R_j$.

We want to fill all the faces $\Phi$ of
$B$. Either $\Phi$ is a hole of $R_0$,
an empty triangle of $G$, or an inner face of some $\boundary R_j$. If $\Phi$ is a hole of $R_0$, then $\Phi$ is outside $R_0$ and then we just copy
in the Kempe nodes from ${\boundary R_0}^\cK_{G-E(R_0)}$. If $\Phi$ is
an inner face of some $\boundary R_j$, we simply copy the configuration of Kempe nodes and main type configuration nodes from this face in ${\boundary R_j}^\cKC_{G R_j}$.
As a final cleaning step, note that we, from different $R_j$, may have configurations of the same main type. In this case, we delete all but one of them.

From $B^\cKC_{GR_0}$, we can easily construct ${\boundary R_0}^\cKC_{GR_0}$ and  $\vec K^{\cC}$, where $\vec K$ consists of the Kempe chains stopping in $R_0$.
To get ${\boundary R_0}^\cKC_{GR_0}$ we essentially just delete all Kempe chains in $\vec K$ from $B^\cKC_{GR_0}$
since these are exactly the Kempe chains that do not intersect 
$\boundary R_0$. Next, inside each inner face of 
$\boundary R_0$, we contract the remaining Kempe components to single Kempe nodes. 
Finally, since we have removed the Kempe chains from $\vec K$, the configuration nodes may now be incident to fewer Kempe chains, and it is the incident Kempe chains that define the type. The types thus get reduced, and several of the previous types may get reduced to the same type. However, we keep only one configuration node of each type, and we delete types that become empty in the sense of not being incident to any remaining Kempe chain. The construction of $\vec K^{\cC}$
is similar, but this time, remove the Kempe chains from  $B^\cKC_{GR_0}$ that are not in $\vec K$.

We let $\boundary (R_0;R_1,\ldots,R_d)^\cKC_{GR}$ denote the combination of all the above information. It provides us
${\boundary R_j}^\cKC_{GR_j}$ for $j=0,\ldots,d$, and it tells us how to translate between the enumerations of Kempe chains and main configuration types in the different ${\boundary R_j}^\cKC_{GR_j}$, as well as for
$\vec K^{\cC}$.

In the following bounds for the combined structure, let $i$ be the level of $R_0$; the total size of the parent and child boundaries is $O(\sqrt{r_i})$. Note that all the above information in $\boundary (R_0;R_1,\ldots,R_d)^\cKC_{GR}$
is coded in planar graphs of total size 
$O(\sqrt{r_i})$, hence represented with $O(\sqrt{r_i})$ bits. 
We claim that we can also compute 
$\boundary (R_0;R_1,\ldots,R_d)^\cKC_{GR}$ in 
$O(\sqrt{r_i})$ time. The only slightly non-trivial issues are sorting the main types. Normally sorting takes super-linear time, but not in our case. Recall that the main type of a configuration $Z$ is represented as the sorted list of some of the Kempe chains from $\boundary (R_0;R_1,\ldots,R_d)^\cK_{G}$ that are touched by $Z$.
These Kempe chains are indexed within $\boundary (R_0;R_1,\ldots,R_d)^\cK_{G}$, so all relevant indices are $O(\sqrt{r_i})$.
The list has constant length because $Z$ can only touch a constant number of Kempe chains (indeed at most 18), so sorting the index list is trivially done in constant time. The total number of main types nodes in the input is $O(\sqrt{r_i})$. We can therefore sort the main types lexicographically in $O(\sqrt{r_i})$ time using the classic Radix Sort. We do this sorting when we want to enumerate and index the main types. We also use it when we want to identify nodes of the same type in order to remove duplicates.
Thus we can compute $\boundary (R_0;R_1,\ldots,R_d)^\cKC_{GR}$ in 
$O(\sqrt{r_i})$ time. As usual, we note that if
$\sqrt{r_i}=O(\ell)$, then we can compute everything in constant time with a table look-up. This, however, does not include counters and sub-types.

\paragraph{Involving subtypes and counters}
We will now describe combinatorially how to deal with subtypes and counters. We are again considering a region $R_0$ with children $R_1,\ldots,R_d$. Inductively, what we get from each
$R_j$ is $\boundary {R_j}^\cKC_{GR_j}$, which has a configuration node of each main type relative to the Kempe chains intersecting $\boundary {R_j}$. Such a configuration node $T$ represents all live configurations $Z$ touching the same subset of
these Kempe chains. When we say that $Z$ is live, we mean that its preferred improving change agrees with all the swap decisions made for Kempe chains stopping below $R_j$. The live configurations represented by $T$ are divided into constantly many subtypes based on their preferred improving swap decisions for the constantly many Kempe chains that they touch and which intersect $\boundary R_j$. These are the only undecided Kempe chains that they touch. If $T$ represents no live configurations, it should have been deleted. For each subtype, we have a counter for how many live configurations it represents. Our enumeration of main types in 
$\boundary {R_j}^\cKC_{GR_j}$ can be easily extended to the subtypes, and then the counters are all contained in an array $A_j$. We get such an array $A_j$ from each child $R_j$, and in the recursive step we fill out $A_0$ for $R_0$.

With $B=\bigcup_{j=1}^d
\boundary R_j$, we will now construct $B^\cKC_{GR_0}$, this time including subtypes with counters. Several child
structures $\boundary {R_j}^\cKC_{GR_j}$ could have the same main type, that is, live configurations touching the same set of undecided Kempe chains. When we delete duplicates, we just add up the counters for the subtypes.

Next we want to construct $\vec K^C$ where $\vec K$ are the Kempe chains stopping in $R_0$. Each main type $T$ from
$\boundary {R_j}^\cKC_{GR_j}$ touched a subset of the Kempe chains intersecting $\boundary {R_j}^\cKC_{GR_j}$. From this set, we get a reduced type $T'$ that removes all Kempe chains not in $\vec K$, so the reduced type only tells which Kempe chains from $\vec K$ it touches.

Now consider the subtypes of $T$. They are determined by the preferred swap decisions for all Kempe chains intersecting $\boundary R_j$ that they touch. They only need preferred swap decisions for the Kempe chains from $\vec K$ that they touch. However, the subtype also counts the number $u$ of undecided Kempe chains touched by the configuration that intersect $\partial R_j$ and are not in $\vec K$.

We now have all the counters needed to make the swap decisions for the Kempe chains in $\vec K$, and we record these decisions.

The next step is to remove Kempe chains in $\vec K$ from $B^\cKC_{GR_0}$. First
we have to remove all subtypes not agreeing 
with the decisions we just made on $\vec K$. 
The remaining subtypes are only distinguished by the decisions on Kempe chains touching $\boundary R_0$, and then we can construct $\boundary {R_0}^\cKC_{GR_0}$ 
as described earlier, once again just adding up counters for subtypes when duplicates are removed. The counters for all types are contained in an array $A_0$, completing the recursive step.

Even including counters of subtypes, a simple algorithm can do the above processing of $R_0$ in time $O(\sqrt{r_i})$. As usual, this suffices if $R_0$ is on level $i$ with $r_i=\Omega((\log n)^2)$.

However, with counters, we must process more information. If $R_0$ is a region on level $i$, then the counters for subtypes may need $\log r_i$ bits, so the total space is $O(\sqrt{r_i}\log r_i)$, and 
if this is $\omega(\ell)$, then we cannot do the processing just using a look-up table. This means that we can only use simple table look-ups for 
$r_i=O((\log n)^2/(\log\log n)^2)$. Below, we handle the remaining cases.

\subsection{Transferring counters.}
We will now show how to handle 
regions on level $i$ where
\[(\log n)^2/(\log\log n)^2<r_i<(\log n)^2.\]
Let $q$ be an integer such that $q\ell=\Theta(\sqrt{r_i}\log r_i)$. Then $q=O(\log\log n)$ is the factor by which the information including counters is too large for a single table look-up. 

However, we will show how to handle this case in $O(q^2)$ time.
There are $O(n(\log\log n)^2/
(\log n)^2)$ such regions, so the total time for all these regions becomes
$O(n(\log\log n)^4/(\log n)^2)=o(n/(\log n))$, which is fine for the proof of Lemma \ref{lem:fast-det-Kempe-change}.

The construction below is elementary, but this is only because we have been very careful above to make sure that we have a compact understanding of the connection between different enumerations of types.

We now consider a region $R_0$ on level $i$ with children $R_1,\ldots,R_d$.
Together with each child $R_j$, we have ${\boundary R_j}^\cKC_{G R_j}$ in $O(\ell)$ bits plus the array $A_j$ of counters for all types in
${\boundary R_j}^\cKC_{G R_j}$. The total processing of $R_0$ is summarized in Table \ref{fig:lookups}.

\begin{table}
\begin{center}
\small
\renewcommand{\arraystretch}{1.35}
\begin{tabularx}{0.9\textwidth}{@{}p{0.30\textwidth}Xp{0.08\textwidth}}
\toprule
Look-up & Additional input  & Iterations \\
\midrule
Construct $\vec K[h]^C$ & The segment $\vec K[h]$ & $O(q)$ \\
Update $B[h]$ & $\vec K[h]^C$, $S$, $A_j[g]$, and current $B[h]$ & $O(q^2)$ \\
Decide swaps for $\vec K[h]$ & $\vec K[h]^C$ and populated $B[h]$  & $O(q)$ \\
Update $A_0[h]$ & Completed $S$, $A_j[g]$, and current $A_0[h]$ & $O(q^2)$ \\
\bottomrule
\end{tabularx}
\end{center}

    \caption{
The table summarizes the look-up steps in the order they are done for the processing of the region $R_0$.
Whenever required, the input also contains
$\partial(R_0;R_1,\ldots,R_d)^{KC}_{GR}$,
the relevant segment indices and lengths, the counter width, and the
length of the decided prefix in $S$. All these inputs, together with
any one source array segment and one target array segment, occupy $O(\ell)$ bits.}
    \label{fig:lookups}
    \end{table}

Let $\vec K$ be Kempe chains stopping in $R_0$. We already know how to compute  $\vec K^\cC$ with a single table look-up, and this includes all the main types relative to $\vec K$. However, to make the swap decisions for $\vec K$, we need all the subtype counters from $A_1,\ldots,A_d$.

The simple idea now is that we partition $\vec K$ into 
$q$ segments $\vec K[h]$ each with $\ell^-=O(\sqrt {r_i}/q)$ Kempe chains, and make the swap decisions for one segment $\vec K[h]$ at a time. The point here is that an array of $O(\ell^-)$ counters, each with $O(\log r_i)$ bits, fits in $O(\ell)$ bits, and this can be used for a table look-up.
As we proceed, we add the decisions to one swap decision vector $S$ for all of $\vec K$. The vector $S$ is stored in $O(\ell)$ bits.

To make the decisions for $\vec K[h]$, we first construct $\vec K[h]^\cC$ describing the configuration types relevant to $\vec K[h]$. By Lemma \ref{lem:Kempe-config}, there are only $O(|\vec K[h]|)=O(\ell^-)$ main types, and only a constant factor more subtypes. We will populate a corresponding counter array $B[h]$, noting that it only needs $O(\ell)$ bits. When 
$B[h]$ has been populated, we can make the swap decisions for $\vec K[h]$ using a single table look-up.

From $\boundary(R_0;R_1,\ldots,R_d)^\cKC_{GR}$ and $\vec K[h]^\cC$, we have all the information we need to map type indices from the $A_j$ to type indices from $B[h]$. Combining this with the vector $S$ of prior swap decisions for $\vec K[1..h-1]$, we also know which subtypes should be eliminated because they did not agree with the prior decision. 

We now also divide $A_j$ into $q$ segments of $\ell^-$ counters.
One such segment $A_j[g]$ is stored in $O(\ell)$ bits. If 
we combine this with $B[h]$, $\boundary(R_0;R_1,\ldots,R_d)^\cKC_{GR}$, $\vec K[h]^\cC$, and $S$, then we know how to update $B[h]$ with the counters from $A_j[g]$. This is only $O(\ell)$ bits of information, so we can compute $B[h]$ updated with $A_j[g]$ using a single look-up.

We have to update $B[h]$ using the counters from all segments in all $A_j$, but there are only $O(q)$ such segments in total, so we need only $O(q)$ table look-ups. When done populating $B[h]$, we make the swap decisions for $\vec K[h]$ as described above using a single table look-up.

We have to repeat this for each of the $q$ segments $\vec K[g]$ of Kempe chains stopping in $R_0$. The total number of table look-ups to compute all swap decisions for $\vec K$ is therefore $O(q^2)$, as claimed.

Next, to complete the recursive step, we need to compute the array $A_0$ of subtype counters for ${\boundary R_0}^\cKC_{G R_0}$. 

We use the exact same idea as above, dividing $A_0$ into segments $A_0[h]$. Using  $A_0[h]$, $A_j[g]$, $\boundary(R_0;R_1,\ldots,R_d)^\cKC_{GR}$, and $S$, we know to update
$A_0[h]$ with $A_j[g]$, and we can again do this with a single table look-up. The total number of segment pairs considered is $O(q^2)$. 
We conclude that we can complete the recursive step for $R_0$
in $O(q^2)$ time, as claimed. 

This completes the proof of 
Lemma \ref{lem:fast-det-Kempe-change}, hence of Theorem \ref{thm:task-linear} and Theorem \ref{thm:linear-time}.

\bibliographystyle{abbrv} 
\bibliography{paper,4CT}

@article {steinberger2010unavoidable,
    AUTHOR = {Steinberger, John P.},
     TITLE = {An unavoidable set of {D}-reducible configurations},
   JOURNAL = {Trans. Amer. Math. Soc.},
  FJOURNAL = {Transactions of the American Mathematical Society},
    VOLUME = {362},
      YEAR = {2010},
    NUMBER = {12},
     PAGES = {6633--6661},
   MRCLASS = {05C15},
  MRNUMBER = {2678989},
       DOI = {10.1090/S0002-9947-2010-05092-5},
       URL = {http://dx.doi.org/10.1090/S0002-9947-2010-05092-5},
}

@article {kempe1879,
    AUTHOR = {Kempe, Alfred B.},
     TITLE = {On the Geographical Problem of the Four Colours},
   JOURNAL = {Amer. J. Math.},
  FJOURNAL = {American Journal of Mathematics},
    VOLUME = {2},
      YEAR = {1879},
    NUMBER = {3},
     PAGES = {193--220},
}

@article {AllaireSwart78,
    AUTHOR = {Allaire, Frank and Swart, Edward Reinier},
     TITLE = {A systematic approach to the determination of reducible
              configurations in the four-color conjecture},
   JOURNAL = {J. Combin. Theory Ser. B},
  FJOURNAL = {Journal of Combinatorial Theory. Series B},
    VOLUME = {25},
      YEAR = {1978},
    NUMBER = {3},
     PAGES = {339--362},
   MRCLASS = {05C15},
  MRNUMBER = {516267},
       DOI = {10.1016/0095-8956(78)90010-2},
       URL = {http://dx.doi.org/10.1016/0095-8956(78)90010-2},
}

@book {AppelHaken89,
    AUTHOR = {Appel, Kenneth and Haken, Wolfgang},
     TITLE = {Every planar map is four colorable},
    SERIES = {Contemporary Mathematics},
    VOLUME = {98},
      NOTE = {With the collaboration of J. Koch},
 PUBLISHER = {American Mathematical Society, Providence, RI},
      YEAR = {1989},
     PAGES = {xvi+741},
      ISBN = {0-8218-5103-9},
   MRCLASS = {05C15},
  MRNUMBER = {1025335},
       DOI = {10.1090/conm/098},
       URL = {http://dx.doi.org/10.1090/conm/098},
}

@article {4ct2,
    AUTHOR = {Appel, K. and Haken, W. and Koch, J.},
     TITLE = {Every planar map is four colorable. {II}. {R}educibility},
   JOURNAL = {Illinois J. Math.},
  FJOURNAL = {Illinois Journal of Mathematics},
    VOLUME = {21},
      YEAR = {1977},
    NUMBER = {3},
     PAGES = {491--567},
   MRCLASS = {05C15},
  MRNUMBER = {0543793},
       URL = {http://projecteuclid.org/euclid.ijm/1256049012},
}

@article {4ct1,
    AUTHOR = {Appel, K. and Haken, W.},
     TITLE = {Every planar map is four colorable. {I}. {D}ischarging},
   JOURNAL = {Illinois J. Math.},
  FJOURNAL = {Illinois Journal of Mathematics},
    VOLUME = {21},
      YEAR = {1977},
    NUMBER = {3},
     PAGES = {429--490},
   MRCLASS = {05C15},
  MRNUMBER = {0543792},
       URL = {http://projecteuclid.org/euclid.ijm/1256049011},
}

@article {birkhoff,
    AUTHOR = {Birkhoff, George D.},
     TITLE = {The Reducibility of Maps},
   JOURNAL = {Amer. J. Math.},
  FJOURNAL = {American Journal of Mathematics},
    VOLUME = {35},
      YEAR = {1913},
    NUMBER = {2},
     PAGES = {115--128},
  MRNUMBER = {1506176},
       URL = {https://doi.org/10.2307/2370276},
}

@article {RSST,
    AUTHOR = {Robertson, Neil and Sanders, Daniel and Seymour, Paul and
              Thomas, Robin},
     TITLE = {The four-colour theorem},
   JOURNAL = {J. Combin. Theory Ser. B},
  FJOURNAL = {Journal of Combinatorial Theory. Series B},
    VOLUME = {70},
      YEAR = {1997},
    NUMBER = {1},
     PAGES = {2--44},
   MRCLASS = {05C15},
  MRNUMBER = {1441258},
       URL = {https://doi.org/10.1006/jctb.1997.1750},
}

@article {Franklin1922,
    AUTHOR = {Franklin, Philip},
     TITLE = {The {F}our {C}olor {P}roblem},
   JOURNAL = {Amer. J. Math.},
  FJOURNAL = {American Journal of Mathematics},
    VOLUME = {44},
      YEAR = {1922},
    NUMBER = {3},
     PAGES = {225--236},
  MRNUMBER = {1506473},
       URL = {https://doi.org/10.2307/2370527},
}

@inproceedings {RSST-STOC,
    AUTHOR = {Robertson, Neil and Sanders, Daniel P. and Seymour, Paul and
              Thomas, Robin},
     TITLE = {Efficiently four-coloring planar graphs},
 BOOKTITLE = {Proceedings of the {T}wenty-eighth {A}nnual {ACM} {S}ymposium
              on the {T}heory of {C}omputing ({P}hiladelphia, {PA}, 1996)},
     PAGES = {571--575},
 PUBLISHER = {ACM, New York},
      YEAR = {1996},
      ISBN = {0-89791-785-5},
   MRCLASS = {05C85 (05C15)},
  MRNUMBER = {1427555},
       DOI = {10.1145/237814.238005},
       URL = {https://doi.org/10.1145/237814.238005},
}

@article{tait1880dual,
  title   = {Note on a theorem in geometry of position},
  author  = {Peter Guthrie Tait},
  journal = {Trans. Roy. Soc. Edinburgh},
  volume  = {29},
  pages   = {657--660},
  year    = {1880}
}

@article{madison1897,
  title   = {Note on the history of the map-coloring problem},
  author  = {Isabel Maddison},
  journal = {Bull. Amer. Math. Soc.},
  volume  = {3},
  pages   = {257},
  year    = {1897}
}

@article{nearlinear,
 title={The Four Color Theorem with Linearly Many Reducible Configurations and Near-Linear Algorithm},
 author={Inoue, Yuta and Kawarabayashi, Ken-ichi and Miyashita, Atsuyuki and Mohar, Bojan and Thomassen, Carsten and Thorup, Mikkel},
 journal={ https://arxiv.org/abs/2603.24880}, 
 year={2026},
 note = {{E}xtended Abstract at FOCS'26.}}

@article{turan1984succinct,
  author    = {Gy{\"o}rgy Tur{\'a}n},
  title     = {On the succinct representation of graphs},
  journal   = {Discrete Applied Mathematics},
  volume    = {8},
  number    = {3},
  pages     = {289--294},
  year      = {1984},
  doi       = {10.1016/0166-218X(84)90126-4},
  url       = {https://doi.org}
}

@article{arlazarov1970economical,
  author     = {Arlazarov, V. L. and Dinic, E. A. and Kronrod, M. A. and Faradzev, I. A.},
  title      = {On economical construction of the transitive closure of a directed graph},
  journal   = {Doklady Akademii Nauk SSSR},
  volume    = {194},
  number    = {3},
  pages     = {487--488},
  year      = {1970},
  publisher = {Russian Academy of Sciences}
}

@article{keeler1995short,
  author    = {Rick Keeler and Jeff Westbrook},
  title     = {Short encodings of planar graphs and maps},
  journal   = {Discrete Applied Mathematics},
  volume    = {58},
  number    = {3},
  pages     = {239--252},
  year      = {1995},
  doi       = {10.1016/0166-218X(93)E0150-W},
  url       = {https://doi.org}
}

@inproceedings{PT06-pred,
  author       = {Mihai P{\u{a}}tra{\c{s}}cu and
                  Mikkel Thorup},
  editor       = {Jon M. Kleinberg},
  title        = {Time-space trade-offs for predecessor search},
  booktitle    = {Proceedings of the 38th Annual {ACM} Symposium on Theory of Computing,
                  Seattle, WA, USA, May 21-23, 2006},
  pages        = {232--240},
  publisher    = {{ACM}},
  year         = {2006},
  url          = {https://doi.org/10.1145/1132516.1132551},
  doi          = {10.1145/1132516.1132551},
  bibsource    = {dblp computer science bibliography, https://dblp.org}
}

@inproceedings{miltersen1996lower,
  author    = {Miltersen, Peter Bro},
  title     = {Lower bounds for static dictionaries on {RAMs} with bit operations but no multiplication},
  booktitle = {Proceedings of the 23rd International Colloquium on Automata, Languages and Programming (ICALP 1996)},
  series    = {Lecture Notes in Computer Science},
  volume    = {1099},
  pages     = {442--453},
  year      = {1996},
  publisher = {Springer},
  doi       = {10.1007/3-540-61440-0_149}
}

@book{kernighan1978c,
  author    = {Kernighan, Brian W. and Ritchie, Dennis M.},
  title     = {The C Programming Language},
  year      = {1978},
  publisher = {Prentice-Hall},
  address   = {Englewood Cliffs, NJ},
  isbn      = {0-13-110163-3}
}

@inproceedings{KMS13,
  author    = {P.N. Klein and S. Mozes and C. Sommer},
  title     = {Structured Recursive Separator Decompositions
for Planar Graphs in Linear Time},
  booktitle = {Proc. 45th STOC},
  year      = {2013},
  pages     = {505--514},
  note      = {Full version in CORR, arxiv:1208.2223},
  url          = {https://doi.org/10.48550/arXiv.1208.2223}
}

@article{Fre87,
  author    = {G.N. Frederickson},
  title     = {Fast Algorithms for Shortest Paths in Planar Graphs, with Applications},
  journal   = {{SIAM} J. Comput.},
  volume    = {16},
  number    = {6},
  pages     = {1004--1022},
  year      = {1987},
  url       = {https://doi.org/10.1137/0216064},
  doi       = {10.1137/0216064},
  bibsource = {dblp computer science bibliography, https://dblp.org}
}
\appendix
\section{Our notion of divisions versus those in \cite{KMS13}}\label{app:translate}
In this appendix, we describe how to translate between our notion of $r$-divisions and those in \cite{KMS13}. Our presentation assumes that the reader is familiar with \cite{KMS13} so that we only have to focus on the differences.

\paragraph{Difference in definitions}
Our definition of divisions is a bit cleaner and stricter than the one in \cite{KMS13}, but is satisfied by the construction in 
\cite{KMS13}. 

In both papers, a \emph{division} of $G$ is a collection of regions $R_1,\ldots,R_k$. In this paper, we require that every face of $G$ is an inner face of exactly one region $R_i$. 

In \cite{KMS13} they only require that every edge in $G$ is in at least one region. They  say a boundary edge \cite{KMS13}
is any edge that is shared by multiple regions. Like us, they further require that such shared edges are on holes in any region they belong to. However, with their definition, an edge on a hole in a region  might not be a boundary edge. An example could be a division into three regions meeting in the 
sides of an inner triangle $\Delta$ of $G$. If none of 
them contain $\Delta$, then 
they each have a side of $\Delta$ on a hole, but these sides are not 
shared, so according to \cite{KMS13} they are not boundary edges.
With our requirements to divisions, the above situation cannot happen because some region would have to contain $\Delta$. As noted in Section \ref{sec:basics}, with our requirements for divisions, an edge $uv$ is on a hole of $R_i$ (our definition of a boundary edge) if and only if $uv$ is shared with another region $R_j$ (the definition of a boundary edge in \cite{KMS13}).

The construction from \cite{KMS13} implies that a vertex shared by regions is incident to a shared edge, 
so, for edges, a vertex is on a hole of a region if and only if it is shared by another region.

A nice feature of defining boundary vertices and edges as those on a hole of a region is that we only need to know a region $R$ to know its boundary $\boundary R$. We 
do not need to know anything about other regions of a division.

\paragraph{The construction of divisions is the same (almost)}

We now consider the recursive bi-division in \cite[\S 3]{KMS13}.

Starting from a region $R$, in step 7, they find a simple-cycle separator $C$. This is a simple cycle in the graph $G'$ where they have placed a star in each hole of $G$. This means the cycle follows edges of $G$ and can cross a hole via the star, but because it is simple, it can cross any given hole at most once.

The cycle partitions the plane, and we have two closed 
discs $\Phi_0=\overline{\textnormal{Int}(C)}$ and $\Phi_1=\overline{\textnormal{Ext}(C)}$. For $i=0,1$, we have a child region $R_i$
consisting of all parts of  $R$ contained in 
$\Phi_i$. In particular, this implies that we get a bipartition of the inner faces of $R$.
Inductively, it follows that if we start from $R=G$ where all faces of $R$ are inner, then the recursive division will partition the inner faces between the resulting regions.

In \cite{KMS13}, they call our inner faces natural faces, and they do indeed point out that the bi-division bi-partitions these faces. However, they make a minor mistake in their 
definition of the child regions $R_i$.
They only include edges on the sides of the inner faces or in $C\cap R$, but as we can see in their Figure 1 (d), we can get edges with holes on both sides, and if $R$ had such an edge $uv$ that is not in $C$, then \cite{KMS13} would not pass $uv$ down to any of its children. This could lead to disconnected regions contradicting their unproven Lemma~4, and it also leads to a contradiction of their Definition 7.
Here we include those forgotten edges in $R_0$ and $R_1$ by giving them \emph{all} parts of $R$ contained in $\Phi_0$ and $\Phi_1$, respectively.

\paragraph{Bound on the number of children}
A final difference is that we state that regions on level $i$ have $O(r_i/r_{i-1})$ children on level $i-1$. This bound is not in  \cite{KMS13}, but it follows by a modification of their proof in \cite[\S 3.3]{KMS13} that
their $r$-division has only $O(N/r)$ regions where $N=|V|$ (while modifying their proof, we also had to fix a gap in their argument, see discussion above Lemma \ref{lem:S-bound} below). Their recursive subdivision is topologically embedded in a binary recursive subdivision represented by a binary tree $\cT$. Each node $x\in\cT$ represents a region $R(x)$. Let $n(x)=|R(x)|$ and $b(x)=|\boundary R(x)|$. For $r<n(x)$, let
$S_r(x)$ be the set of the highest descendants $y$ of $x$ such that $n(y)\leq r$ (possibly $y=x$). The regions represented by $S_r(x)$ form a division, but not yet an $r$-division.
Next let 
$S'_r(y)$ be the set of the highest descendants $z$ of $y$ such that $b(z)\leq c'\sqrt r$ for some constant $c'$. Finally define
$S^*_r(x)=\bigcup\{S'_r(y)\mid {y\in S_r}(x)\}$.
Relating back to our levels, if $R(x)$ is a region on level $i$ then
its child regions on level $i-1$ form the set 
$S^*_{r_{i-1}}(x)$. Moreover, by definition of level $i$, $n(x)\leq r_i$ and $b(x)=O(\sqrt{r_i})$. 

For general $x\in\cT$ and $r\leq n(x)$, we will argue that $|S^*_r(x)|=O(n(x)/r+b(x)/\sqrt r)$. 
In \cite[\S 3.3]{KMS13} they
argue that if $x$ is
on depth $0\pmod 3$ in $\cT$,
then  $\sum_{y\in S_r(x)} b(y)\leq
b(x)+O(n(x)/\sqrt r)$. Working on those depths is all they need in \cite[\S 3.3]{KMS13}, but the same holds for arbitrary depths since the binary subdivision uses separator cycles of size at
most $O(\sqrt{n(x)})=O(n(x)/\sqrt r)$.

Next we want to argue that $|S_r(x)|=O((n(x)+b(x))/r)$. In \cite{KMS13}, they argue this (for the special case where $x$ is the root with $b(x)=0$) by referring to the parents of $y$ of 
$z\in S_r(x)$, since they have $n(y)>r$. However, these parents can overlap a lot, e.g., $y$ could have another child $y'$ containing almost all of $y$, and with a child $z'\in S_r(x)$.
We therefore present a different proof that fixes this issue in \cite{KMS13}.
\begin{lemma}\label{lem:S-bound}
$|S_r(x)|=O((n(x)+b(x))/r)$.
\end{lemma}
\begin{proof}
We want to argue that a constant fraction of the nodes $z\in S_r(x)$ have $n(z)>r/2$.

Consider any node $y\in\cT$ with children $y_0$ and $y_1$. We assume $n(y)$ is bigger than some large enough constant. From the construction in \cite{KMS13}, we have some simple inequalities. Since 
$R(y_0)$ and $R(y_1)$ divide $R(y)$ with some overlap, we have
$n(y_i)\leq n(y)$ and 
$n(y_0)+n(y_1)>n(y)$.
Moreover, if $y$ is at depth $0 \bmod 3$ in $\cT$, then by \cite[(2)]{KMS13},
\begin{equation}\label{eq:max}\max\{n(y_0),n(y_1)\}
    \leq 0.75 n(y)+O(\sqrt{n(y)})<0.8 n(y).
\end{equation}
This also implies
\begin{equation}\label{eq:min}\min\{n(y_0),n(y_1)\}> 0.2 n(y).
\end{equation}
Now consider the subtree $T$ of $\cT$ rooted in $x$ and containing all nodes $y$ with $n(y)>r$. The parents of nodes $z\in S_r(x)$ are exactly the nodes in $T$ with fewer than two children. Also, $|S_r(x)|\leq 2|T|$. 
Moreover, if $y$ is a leaf in $T$, then it has at least one child $z\in S_r(x)$ with $n(z)>r/2$.

We need to argue that the size of $T$ is linear in the number of leaves. To do that, we will bound the maximal length of a path of single-child-nodes in $T$. 

Every third node $y$ on such a path is at depth $0\bmod 3$, and if $n(y)\geq 5r$, then by \req{eq:min}, the children have $n(y_i)>0.2 n(y)\geq r$, so both children are in $T$, that is, $y$ is a branching node.

However, if $n(y)<5r$ and $y$ is at depth $0\mod 3$, then, by \req{eq:max}, both children have $n(y_i)\leq 0.8 n(y)$. This means that after passing at most 8 such nodes, all children
have $n(y_i)\leq r$, implying that they are not in $T$, i.e., that we have reached a leaf.
We conclude that the maximal single-child-path in $T$ is of length 24, hence that the total number of nodes in $T$ is linear in the number of leaves. This completes the proof that a constant fraction of the nodes $z\in S_r(x)$ have $n(z)>r/2$.

Since only boundary vertices are shared in the division represented by $S_r(x)$, we know that
\[\sum_{z\in S_r(x)}n(z)\leq n(x)+\sum_{z\in S_r(x)}b(z)\leq n(x)+b(x)+O(n(x)/\sqrt r)=O(n(x)+b(x)).
\]
Since a constant fraction of the nodes $z\in S_r(x)$ have $n(z)>r/2$, we conclude that 
$|S_r(x)|=O((n(x)+b(x))/r)$.
\end{proof}
By \cite[(13)]{KMS13}, for every $y\in S_r(x)$, we have $|S_r'(y)|=O(1+b(y)/\sqrt{r})$, so we conclude that 
\[|S^*_r(x)|=\sum_{y\in S_r(x)}|S_r'(y)|=O\left(|S_r(x)|+
\frac{b(x)+O(n(x)/\sqrt r)}{\sqrt r}\right)=O(n(x)/r+b(x)/\sqrt r).\]
Finally, we insert $n(x)\leq r_i$, $b(x)=O(\sqrt{r_i})$, and $r=r_{i-1}$, and get
$|S^*_{r_{i-1}}(x)|=O(r_i/r_{i-1}+\sqrt{r_i}/\sqrt{r_{i-1}})=O(r_i/r_{i-1})$.

\end{document}